\documentclass[11pt,a4paper]{article}

\usepackage[english]{babel}
\usepackage[letterpaper,margin=1in]{geometry}

\usepackage[T1]{fontenc}
\usepackage{libertine}

\usepackage{setspace}
\usepackage{array}
\usepackage{multirow}
\usepackage{graphicx}
\graphicspath{{images/}}

\usepackage{subcaption}
\usepackage{caption}

\usepackage{amsthm}
\usepackage{amsmath}
\usepackage{amssymb}

\usepackage{xcolor}
\usepackage{color}
\usepackage{comment}
\usepackage{soul}
\usepackage[normalem]{ulem}

\usepackage{microtype}
\usepackage{lineno}

\usepackage{hyperref}
\hypersetup{
    colorlinks=true,
    allcolors=blue
}

\usepackage{authblk}

\newtheorem{theorem}{Theorem}[section]
\newtheorem*{Statement 1}{\bf Statement 1}
\newtheorem*{Statement 2}{\bf Statement 2}
\newtheorem{lemma}[theorem]{Lemma}
\newtheorem{proposition}[theorem]{Proposition}

\newtheorem{Rem}{Remark}

\newtheorem{Problem}[theorem]{\bf Problem}

\title{Optimizing Line Segment Inspection with\\
Limited-Range Drones}

\author[1]{José-Miguel Díaz-Báñez}
\author[1]{José-Manuel Higes\footnote{Corresponding author: jhiges@us.es}}
\author[1]{Alina Kasiuk}
\author[1]{Inmaculada Ventura}

\affil[1]{Department of Applied Mathematics, University of Seville, Spain}

\date{}

\begin{document}

\maketitle

\begingroup
\renewcommand{\thefootnote}{}
\footnotetext{
This work is supported in part by grants
PID2024-159968OB-I00 and TED2021-129182B-I00
funded by MCIN/AEI/10.13039/501100011033 and
the European Union Next GenerationEU/PRTR.
}
\addtocounter{footnote}{-1}
\renewcommand{\thefootnote}{\arabic{footnote}}
\endgroup

\begin{abstract} Optimization problems with drones are widely studied in a variety of civilian tasks, mainly due to their ability to traverse rough terrains and to carry cameras and other sensors for surveillance tasks. The limited battery life of these aerial robots poses challenges in operational research.
	In this paper, we address the following optimization problem.
	We are given a set of line segments (e.g. tubes in a solar plant) to inspect by drones. The objective is to detect broken pipes using artificial intelligence and path planning must be carried out efficiently. On the one hand, the limited capacity of the batteries necessitates periodic visits (tours) to a fixed base station. However, it is desirable to allocate a set of tours for each drone to ensure that the segments are covered as quickly as possible, aiming to minimize the makespan, which is the maximum time spent by any drone. We are able to prove that this optimization problem is strongly NP-hard even when the segments are positioned on a line and the scenario involves only two drones. Then,  approximation algorithms are proposed.
	Our computational experiments demonstrate that the proposed algorithm achieves near-optimal performance across diverse operational scenarios.
\end{abstract}




\maketitle

\section{Introduction}

As technology advances, unmanned aerial vehicles (UAVs), commonly referred to as drones, are assuming an ever-expanding role in the inspection of industrial structures. For example, the manual inspection of high-voltage power transmission lines, gas pipelines or solar plants are both time-consuming and expensive \cite{togola2020real,wanasinghe2020unmanned,fahmani2023unmanned}. Hence, the use of drones equipped with cameras enables efficient fault detection.

Our problem is inspired by the inspection of a particular case of Concentrated Solar Power plants (CSP), which represent a growing technology for electricity production through renewable energies. 
A CSP
plant with Parabolic Trough Collector comprises an array of receiver tubes (where the sun rays are reflected), which are subjected to high thermal stress. Figure \ref{fig:CSP1} illustrates the receiver tubes of a CSP with Parabolic Trough Collectors.
\begin{figure}[ht]
\centering
\includegraphics[scale=0.25]{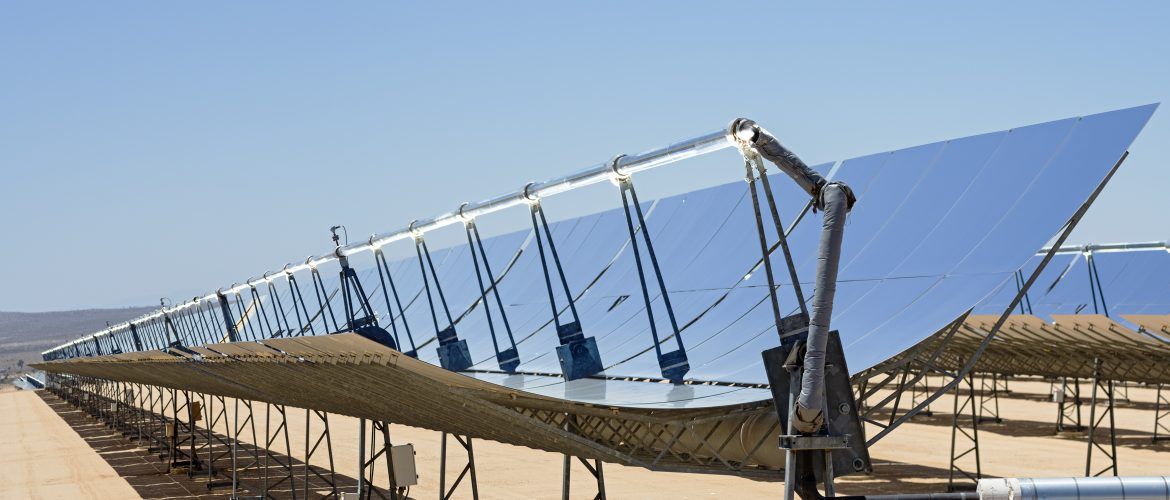}
\caption{Receiver tubes in a CSP-type solar plant.}
\label{fig:CSP1}
\end{figure}
Therefore, the identification of broken glass envelopes is crucial to maintain the proper functioning of the CSP plant \cite{prahl2013airborne,perez2023detecting}. 
In this context, promptly detecting a fault enables the company to take immediate action in the repair process. 
In this paper, we address the problem of minimizing the time required for a team of drones (the completion time or makespan), with limited battery endurance, to efficiently traverse a set of aligned tubes, represented as segments along a straight line.
This problem falls within the domain of arc routing problems (ARPs), an active area in operations research, which involves determining a set of tours with the minimum total cost while traversing a set of links (arcs or edges in a graph), named \emph{required links} \cite{macrina2020drone}. 
In \cite{campbell2018drone}, the key differences from classical non-drone ARPs are discussed. A drone has the flexibility to enter a line through any of its points, traverse a portion of that line, exit through another of its points, subsequently travel directly to any point on another line, and so forth.
Moreover, unlike the vehicles in classical ARPs, which have to follow the links of a given graph, drones can fly directly between any two points. Hence, employing drones for service in ARPs introduces substantial modifications to the conventional methods of modeling and solving these problems. The authors study the following optimization problem for a single drone. Given a depot, a set of (straight or curved) lines to traverse, and a maximum route length $L$, the objective is to partition the lines into feasible routes that start and end at the depot, minimizing the total cost. The cost of a tour includes both the Euclidean distances to reach the lines and the costs associated with traversing the required lines. Obviously, the problem is NP-hard and they present a heuristics approach to solve the problem. After that, in a recent paper \cite{campbell2021solving}, the same authors propose two heuristic methods for the generalized version for $k$ drones. In both papers, the required lines are digitized by approximating each curved line with a dense set of points, effectively representing each line as a polygonal chain. Drones are restricted to entering and exiting these lines only at the designated points along the polygonal chain, transforming the problem into a discrete optimization problem on a graph. Clearly, with this methodology, a large number of points are needed to obtain an accurate set of routes.  

An interesting question that arises here is whether the original continuous problem could be effectively solved in simpler cases, such as when the target segments lie along parallel straight lines. This scenario is relevant for applications such as fault detection in CSP plants, as well as in the inspection of pipelines and power lines.
In a recent paper \cite{bereg2024covering}, the authors provided an affirmative answer for the case where the segments are arranged in one dimension. They developed a polynomial-time algorithm to solve the one-drone/one-line problem, which involves computing a set of tours for a single drone. Given a set of segments positioned along a straight line and a base station located off the line, the algorithm finds tours with a maximum length $L$ that start and end at the base. The goal is to minimize the total time required for the drone to cover all segments.
The drone has limited battery endurance, maintains a constant flying speed, and returns to a base station when its battery is running low. The time spent at the station is disregarded, as it is assumed that battery changes are conducted during this period. Thus, the total time is defined as the sum of the lengths of the necessary tours.  
Additionally, they addressed the problem of minimizing the number of tours required to ``cover" all segments, rather than minimizing the total length. Notably, these two optimization objectives lead to different solutions.

In this paper, we prove that the two-drone/one-line version cannot be solved in polynomial time, unless P = NP. Moreover, we prove that the problem is Strongly NP-Hard and, as a result, no pseudo-polynomial algorithm can exist for solving them unless P = NP. After that, we propose a simple approximation algorithm and demonstrate its effectiveness through both theoretical analysis and experimental evaluation. Our approach uses the exact solution for the one-drone/one-line problem from \cite{bereg2024covering}.
The relevance of solving this problem is two-fold:
from a practical perspective,
minimizing the maximum distance (time) traveled by a team of drones, -essentially, reducing the time required for the team to cover the segments- is meaningful when the company aims to expedite the task and promptly repair the broken tubes. On the other hand, from a methodological point of view, designing  straightforward yet effective algorithms in one dimension can be used to solve more complex instances, for instance,  multiple segments on parallel lines, a structure commonly found in solar power plants. Currently, the company uses a line-by-line approach for drone flights, though not optimized for each line. This bottom-up methodology—starting with efficient exact algorithms for simpler cases and using these as building blocks to design effective approximations for more complex instances—could also prove valuable for tackling other challenging optimization problems in operations research.

The remainder of this paper is organized as follows. In Section \ref{sec:rw}, we discuss various related problems tackled within the area of UAV optimization.
In Section \ref{sec:prob}, we present the formal definition of the problems as well as some interesting remarks. The NP-hardness proof is outlined in Section \ref{section:Np-Hard}, while Section \ref{section: Aproximation algorithm} contains a constant-factor approximation. A MILP formulation and the computational results are described in Section \ref{sec:comp}. Finally, some conclusions and future research directions are drawn in Section~\ref{sec:conclu}.

\section{Related Work}\label{sec:rw}

As mentioned earlier, this work is part of a growing body of research focused on UAV flight optimization algorithms for various civil applications. The challenges in these contexts are diverse: optimizing UAVs locations to minimize their interference \cite{Arribas2020}, determining the positions of the bases to minimize the number of UAVs to cover a region \cite{Chung2021,bereg2022optimal}, or finding the optimal route for minimizing time or other cost functions. For a comprehensive review on optimization models and methods see \cite{CHUNG2020} and \cite{dukkanci2024facility} for drone delivery problems. 

There is a vast body of work on path planning, where drones are typically required to visit only the nodes of a graph, often formulating the problem as a variant of the Traveling Salesman Problem (TSP), such as the Asymmetric Salesman Problem in \cite{Sundar2014} or the Multi-Tour Salesman Problem in \cite{Nekovar2021}. Multi-depot vehicle routing problems for multiple unmanned aerial vehicle operations have been recently addressed using payload capacity constraints  allowing for multiple deliveries during a single trip \cite{uhm2022vehicle,xia2023joint}.
Studies more closely aligned with ours have been inspired by the inspection of power transmission lines \cite{Cui2017,Nekovar2021}. These studies utilize arc routing and path planning methodologies.
However, their problem differs from ours in that, 
the drone covers each segment in a single tour before returning to recharge. 
In contrast, our approach enables segments to be inspected over multiple tours, with UAVs having the flexibility to enter and exit a segment at any point. This results in a more efficient overall coverage.
Allowing drones to enter and exit the lines at any point has also been considered in recent works, such as \cite{campbell2018drone} and \cite{campbell2021solving}. This condition results in an interesting continuous optimization problem with infinitely many feasible solutions. As mentioned above, in these works the lines (whether straight or curved) to be covered are modeled as polygonal chains, with drones allowed to enter and exit each segment only at the waypoints of the polygonal chain. In this way, the drones can travel directly between any two waypoints and a complete graph is generated in which  costs are given by the Euclidean distances. In both papers, the authors present mathematical formulations and propose competitive heuristics.

Typically, algorithms derived from mathematical programming formulations are used in this area \cite{CAVANI2021}.
In fact, it was proved that this type of problems cannot be approximated by constant factors 
with polynomial-time algorithms \cite{Morandi2023}. 
However,
for the one-dimensional case, the authors in \cite{bereg2024covering} proved that minimizing the flight time to cover segments on a line with tours of limited length and a single drone can be solved in polynomial time. This raises the question of whether the problem remains polynomially solvable for two or more drones. In this paper, we prove that it does not and present efficient approximation algorithms.

Additionally, research on drone flights for solar plant inspections has recently gained attention, as in \cite{Jemmali2023}. Similar to our approach, their algorithms also aim to minimize the flight time of multiple drones for inspecting solar power plants, but the problem is framed as a scheduling issue, with pre-established drone routes and a focus on battery decay constraints. In contrast, our problem does not assume predefined tours; instead, we focus on finding the optimal set of routes to traverse the given line segments.

\section{Problems Statement}\label{sec:prob}
Let $\alpha$
be a set of disjoint line segments on a straight line. 
We are given a team of $k$ identical drones that must traverse all the segments in $\alpha$. The drones are constrained by finite battery endurance, maintain uniform velocity during flight, and execute takeoff and battery recharging protocols at a fixed point $B$,  the designated base station located outside the line. We assume that the recharging time is negligible as the batteries are replaced instantly.
Let $L > 0$ denote the maximum distance a drone can travel when initiating and concluding its trajectory (tour) at the fixed point $B$. We consider, without loss of generality, that the orthogonal projection of $B$ to the line is the origin of coordinates $O$. 
For a tour labeled as $t$, we refer to its length as $\ell(t)$. The length of a collection of tours, denoted as $T = \{t_1,\cdots, t_m\}$, is represented as $\ell(T)$. It is calculated as the sum of the lengths of individual tours, that is, $\sum_{i=1}^m \ell(t_i)$. Our objective is to determine a set of tours for each drone in such a way that the maximum length traveled by any drone (the makespan) is minimized while ensuring that all segments are covered. Formally,
this \emph{Minmax problem} for $k$ drones can be stated as follows:
\begin{Problem}\label{Problem: Minmax-k}
$k$-Min-Makespan Problem:
Given a set $\alpha$ of $n$ line segments on  a straight line and $k$ drones, compute a set of tours with maximal length $L$  for each drone, $T_1,T_2, \cdots, T_k$, such that:
\begin{eqnarray}
\notag
\alpha \subset T_1\cup T_2\cup \cdots \cup T_k \,\, \text{\emph{and,}}
\max\limits_{j=1,\cdots,k} 
\ell(T_j) \, \text{\emph{is minimized}}.
\end{eqnarray}
\end{Problem}

See Figure~\ref{fig:definition} for an example of a set of tours covering a set of line segments with 2 drones.

\begin{figure}[ht]
\centering
\includegraphics[scale=0.6]{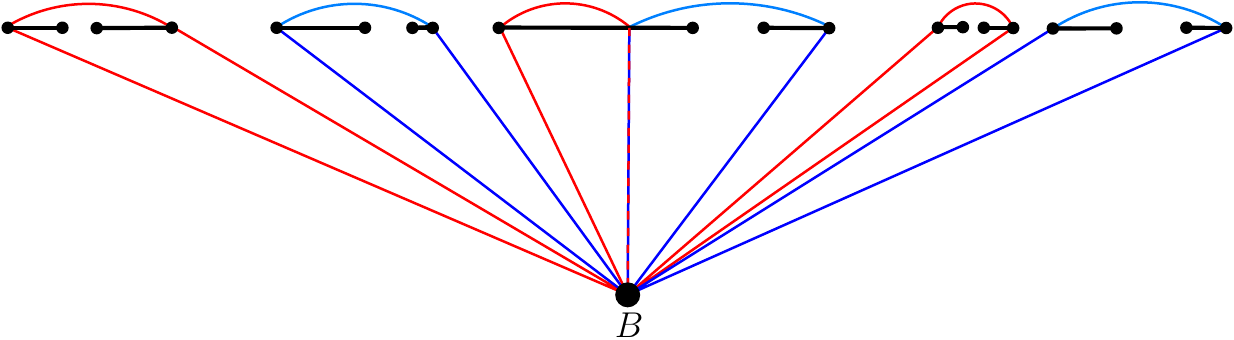}
\caption{Covering tours for two drones.
Blue and red tours correspond to different drones. }
\label{fig:definition}
\end{figure}

\begin{Rem}
We can consider a version of the previous problem when $k = 1$, i.e. minimizing the total length of a set of tours performed by a single drone. Following \cite{bereg2024covering}, we will call such problem \emph{Minsum problem}. 
In \cite{bereg2024covering} it was proved that such a problem can be solved in polynomial time.
\end{Rem}
 
The problem of minimizing the total length by a set of $k$ drones can be trivially solved by allowing one drone to perform all  tours while the others remain inactive.\par
Also, the problem of minimizing the maximal number of tours (instead of total length) performed by the $k$ drones can be easily addressed using the approach of \cite{bereg2024covering}. 

\begin{Problem}\label{Problem: Mintour}
Mintour-$k$: Given a set $\alpha$ of $n$ line segments on  a straight line and $k$ drones, compute a set of tours for each drone, $T_1,T_2, \cdots, T_k$, such that:
\begin{eqnarray}
\notag
\alpha \subset T_1\cup T_2\cup \cdots \cup T_k \,\, \text{\emph{and,}}
\max\limits_{j=1,\cdots,k} 
|T_j| \, \text{\emph{is minimized}}.
\end{eqnarray}
\end{Problem}

\begin{proposition}
    Mintour-$k$ can be solved in polynomial time.
\end{proposition}
\begin{proof}
    Using the greedy approach  of \cite{bereg2024covering}, compute  
    in polynomial time a set of tours $T = \{t_1, \dots, t_m\}$ of minimum cardinality to cover  $n$ segments with just one drone. Distribute the tours of $T$ evenly among each of the drones. Notice that this solution is optimal for Problem \ref{Problem: Mintour}, otherwise we could improve the solution $T$ for one drone. Thus,  $\lceil\frac{m}{k}\rceil$ is the maximum cardinality of tours for each drone and the problem can be solved in $O(n+m)$ time.   
\end{proof}
In this paper, we prove that transitioning from the Minsum problem 
to the Minmax criterion for two drones results in an NP-hard problem. It will be clear from the proofs that our results for two drones can be extended easily for the $k$-Min-Makespan problem. 
Then,  approximation algorithms capable of obtaining the optimal solution for a large majority of scenarios are proposed.

\section{NP-Hardness}\label{sec:NP}\label{section:Np-Hard}
In this section, we will prove that the $2$-Min-Makespan problem is strongly NP-Hard by connecting it with a particular version of \emph{ Finite Partition Problem}.
\subsection{The Finite Partition Problem}\label{subsection: Finite Partition Problem}
It is known that the \emph{two way-balanced finite partition problem} is NP-complete \cite{garey1997computers}: 
\vspace{.25cm}

\emph{Given an even finite set of positive integers $S = \{s_i\}_{i=1}^{2n}$, determine if there is a subset $A\subset S$ of cardinality $n$, with $\sum_{s_i \in A} s_i= \sum_{s_j \in S\setminus A} s_j$.}
\vspace{.25cm}

Furthermore, \cite{WojtczakRationals} established that the finite partition problem for positive rationals is strongly NP-complete, from which we immediately derive: 
\begin{lemma}
    The two way-balanced finite partition problem for positive rationals is strongly NP-complete.  
\end{lemma}
\begin{proof}
    Let $S = \{s_i\}_{i= 1}^{r}$  be a set of positive rationals. Without loss of generality, we can assume $r = 2n$, otherwise we can consider the set $S' = \{s_i'\}_{i=1}^{2n}$ with $s_i' = s_i$ if $i < 2n$ and $s_{2n} = 0$. Now the finite partition problem for $S$ with $|S| = 2n$ can be decomposed in $n$ two way-balanced finite partition problems applied to the $n$ sets  $S_j$ from  $j \in \{0, \cdots, n-1\} $ where $S_j = \{s_i'\}_{i=1}^{2n+2j}$ where $s_i' = s_i$ if $i\le 2n$ and $s_i' = 0$, otherwise. It is clear that if there exists a solution for some $S_j$ to the two-way balanced finite partition problem then there is a solution for the finite partition problem for $S$. Reciprocally, if there is no solution to the finite partition problem for $S$ then there is no solution for any $S_j$ to the two-way balanced finite partition problem. 
\end{proof}

\vspace{.25cm}
We are interested in the optimization version of the problem: 
\begin{Problem}\label{Problem: balanced min max finite partition problem}
Minmax partition problem: Given a finite set of positive integers $S = \{s_i\}_{i=1}^{2n}$, determine 
    a subset $A\subset S$ of cardinality $n$, with $M_2 = \sum_{s_i \in A} s_i\ge M_1= \sum_{s_j \in S\setminus A} s_j$, 
    such that $M_2$ is minimum for all possible subsets of $S$ of cardinality $n$.
\end{Problem}

\begin{lemma}\label{lemma: NP-hardpartition for minmax} Problem \ref{Problem: balanced min max finite partition problem} is Strongly NP-hard. 
\end{lemma}
\begin{proof}
Given a finite set of positive rationals $S = \{s_i\}_{i=1}^{2n}$, suppose that we want to solve the two-way balanced finite partition problem for $S$ and we have an algorithm in polynomial time that solves Problem \ref{Problem: balanced min max finite partition problem}. Then we have a subset $A\subset S$, with $M_2 = \sum_{s_i \in A} s_i$ and $M_1= \sum_{s_j \in S\setminus A} s_j$, $M_2 \ge M_1$, such that $M_2$ is minimum for all possible subsets $A$ of cardinality half of the cardinality of $S$. Assume that $M_2-M_1$ is not minimum. Then there is $M_2'$ and $M_1'$ for another partition $A'$ such that $M_2'-M_1' < M_2-M_1$. By the minimality of $M_2$, we have $M_2'- M_2\ge 0$. Also, we have $M_2 + M_1 = M_2'+M_1'$, therefore $M_1-M_1' \ge 0$. Hence $M_2-M_1 > M_2'-M_1' \ge M_2-M_1$, a contradiction. Therefore, $M_2-M_1$ is the minimum. By minimality, if $M_2-M_1 \ne 0$ then there is no solution for the two-way balanced finite partition problem and, if $M_2 -M_1 = 0$, we have found a solution.  
\end{proof}

The following result will serve as a crucial tool in our construction to prove the Strongly NP-hardness of the $2$-Min-Makespan problem: 
\begin{proposition}\label{Proposition: KC-proposition}
    \emph {For a set of positive numbers $S = \{s_i\}_{i=1}^{2n}$ 
    and two positive constants $K, C >0$, define the set $S' = \{s_i'\}_{i= 1}^{2n}$ with $s_i' = K s_i + C$. Then a subset $A \subset S$ is a solution to Problem \ref{Problem: balanced min max finite partition problem} for $S$ if and only if the subset $A'$, defined by $s_i' \in A'$ with $s_i \in A$, is a solution to Problem \ref{Problem: balanced min max finite partition problem} for the set $S'$.}
\end{proposition}
\begin{proof} 
Assume $A \subset S$ is a solution of Problem \ref{Problem: balanced min max finite partition problem}; then  $ M_2 = \sum_{i \in A} s_i$ is minimum from all possible subsets of $S$ of cardinality $n$ with $M_2 \ge M_1 = \sum_{i \in S\setminus A} s_i$. By linearity, this will imply that the sum $M_2' = K \sum_{i \in A} s_i + C\cdot n$ is minimum from all possible subsets of $S$ of cardinality $n$ with $M_2' \ge M_1' = K\sum_{i \in S\setminus A} s_i + C\cdot n$. Thus, we can conclude that the set $A'$ is a solution to Problem \ref{Problem: balanced min max finite partition problem}.
The reverse implication can be proven in a similar manner.
\end{proof}

\subsection{The construction}\label{section: Construction}
Given a set of positive rationals $S = \{s_i\}_{i= 1}^{2n}$, we will construct a set of segments $\alpha_S$ such that solving the $2$-Min-Makespan problem for $\alpha_S$ will be equivalent to solving Problem \ref{Problem: balanced min max finite partition problem} for $S$. \par
We consider a horizontal line and a constant $L > 0$. Let $B$ be a point at a distance $L/2$ from the line, and let $O$ (the origin of coordinates) be the orthogonal projection of $B$ onto the line.
Furthermore, let $C$ be a point on the horizontal line to the right of $O$, at distance $L/2$ from $B$.
(see Figure \ref{fig:fig14}). Let $S=\{s_i\}_{i = 1}^{2n}$ be a set of positive rationals and assume, without loss of generality, $R=\max_i\{s_i\}=s_{2n}$.  Take $\epsilon$ a small number less than $1/3$. We will say that a point $x$ is a \emph{right-hand point} if it lies on the horizontal line to the right of $O$, at distance $x$. The symmetric point of $x$ with respect to the coordinate axis, noted by $-x$, is a \textit{left-hand point}. 
%
The following steps allow us to construct $\alpha_S$ in such a way that the endpoints of the segments in $\alpha_S$ are either right-hand or left-hand points, i.e.,  $\alpha_S=\left\{a_1, \cdots, a_{2n}\right\}$, where each segment $a_i$ takes the form $[x_i, y_i]$ or $[-y_i, -x_i]$ 
with  $x_i, y_i > 0$ for each $i=1,\cdots, 2n$. Additionally, the inequality $|x_1|<|y_1|<|x_2|<|y_2|<\dots <|x_{2n}|<|y_{2n}|$ must hold.

\begin{itemize}
\item {\emph{Step 1:}} For $i=0,\cdots,2n$, determine a sequence of right-hand points $y_i$, with $c_i = d(y_i, B)$, 
that satisfies:
\begin{equation}
c_0 > \max \left(\frac{L}{2}- \frac{\epsilon L}{4n}, L-y_0-\frac{L}{3}, \frac{\sqrt{5}L}{6}\right),
\end{equation} 
\begin{equation} 
y_{i+1}-y_i+c_i+c_{i+1} = L,
\end{equation}
\begin{equation}
c_{i+1}^2 =  y_{i+1}^2+\left(\frac{L}{3}\right)^2.
\end{equation}
\begin{figure}[ht]
    \centering
    \includegraphics[width=0.8\textwidth,page=4]{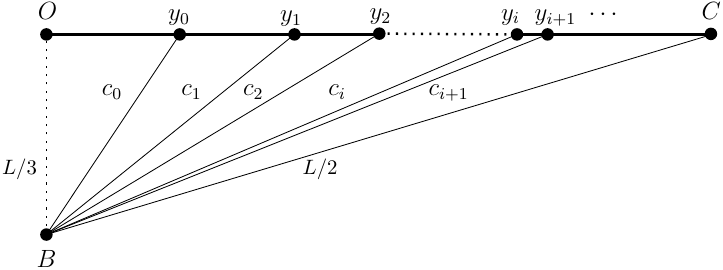}
    \caption{Diagram for Step 1. Right-hand points $y_i$.}
    \label{fig:fig14}
\end{figure}
\item {\emph{Step 2:}} Determine the number sequence: 
\begin{equation}
\notag
    \{s_i' = Ks_i + 2c_{2n}\}_{i=1}^{2n}\\
    \text{ with }
    K = \frac{L-2c_{2n}}{R}.
\end{equation}
\item {\emph{Step 3:}} For $i=1,\cdots,2n$, determine the sequence of right hand points $x_i$, with $b_i = d(x_i, B)$
that satisfies:
\begin{equation}
c_i+(y_i-x_i) + b_i = s_i',
\end{equation}
\begin{equation}
b_i^2 = x_i^2 + \left(\frac{L}{3}\right)^2.
\end{equation}

\item {\emph{Step 4:}} Let $\alpha_S = \{a_i\}_{i=1}^{2n}$ where $a_i = [x_i, y_i]$ if $i$ even, and $a_i = [-y_i, -x_i]$ if $i$ odd (Figure \ref{fig:fig16}).

\begin{figure}[ht]
   \centering
\includegraphics[width=1\textwidth]{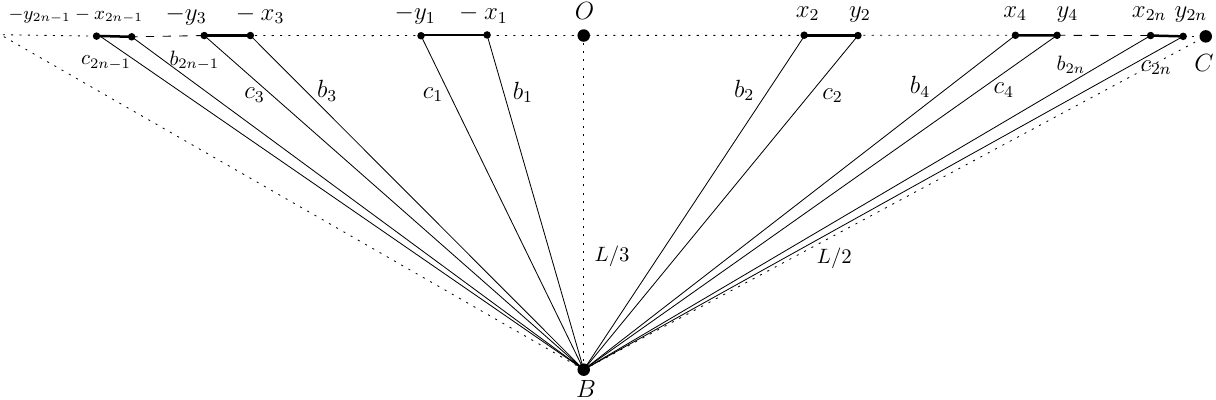}
    \caption{Diagram for Step 4. Final construction.}
    \label{fig:fig16}
\end{figure}
\end{itemize} 

It is easy to see that the construction of $\alpha_S$ can be done in polynomial time. This result follows from the following remark and lemma. 
\begin{Rem}
 Notice that each $c_i$ can be obtained in constant time. Moreover, as in each step we compute $y_i$ by making a tour of length $L$ by equation (2), and the point $C$ is at distance $L/2$ from $B$ then all $y_i < C$ as $C$ is the (infinite) limit of such process. 
 \end{Rem}
 Next lemma proves the case $c_0$, the others are left to the reader. 
\begin{lemma}
    $c_0$ and $y_0$ can be obtained in constant time
\end{lemma}
\begin{proof}
    First determine $y_0'$ and $c_0'$ by making just one tour of length $L$, i.e. $c_0' + L/3+y_0' = L$ . Given that $C$ is the right point of the line at  distance of $L/2$ from $B$, the tour can not reach $C$.    
    Compute now in constant time $y_0''$ as the point corresponding to $c_0'' = \frac{\sqrt{5}L}{6}$ and $y_0'''$ the point corresponding to $c_0''' = \frac{L}{2}- \frac{\epsilon L}{4n}$. It is also clear that both points, $y_0''$ and $y_0'''$ can not reach $C$. Let $y_0^{4}$ be the maximum of $y_0'$, $y_0''$ and $y_0'''$, define, for example, $y_0 = \frac{y_0^{4}+ C}{2}$. By construction $y_0$ can be obtained in constant time and satisfies the properties of Step $1$ 
\end{proof}
Without loss of generality, we will consider a tour to the right of O is traversed in a clockwise direction, that is, the tour begins at the point of the segment closest to $O$ (beginning point) and ends at the point of the segment furthest to $O$ (ending point).  
\begin{lemma}
All the equations of Step 3 can be solved in linear time.     
\end{lemma}
\begin{proof}
First notice that, since $s_i \in (0, R]$  for $i=1,\cdots,2n$, $s_i' \in (2c_{2n}, L]$. For the segment $[y_{i-1}, y_{i}]$ the tour that begins in $y_{i-1}$ and ends in $y_{i}$ has length $L$ (this follows from the second equality of Step 1). The (length)  function $L_i(x)$ that associates to $x \in [y_{i-1}, y_{i}]$ the corresponding length $L_i(x)$ of the tour that begins at $x$  and ends at $y_{i}$ is a continuous decreasing function whose range varies from $2c_i$ to $L$ when $x \in [y_{i-1}, y_i]$. This follows from the fact that the function $L_i(x)$ is defined by: 
\begin{equation}
\notag
L_i(x) = \sqrt{x^2+ \left( \frac{L}{3}\right)^2}+y_i-x+c_i,
\end{equation}
and its derivative is always negative.
Therefore, since each $c_i < c_{2n}$, for each number $s_i'$ between $(2c_{2n}, L]$ we can find  a  $x_i \in [y_{i-1}, y_{i}] $ such that the corresponding tour will have length $L_i(x_i) = s_i'$.
\end{proof}
We now highlight two useful properties of our construction: 
\begin{lemma} \label{lemma: Condition of cardinality}
    Condition of step 1 applied to $c_{2n}$ implies:
    \begin{equation}
        2n(L-2c_{2n}) < \epsilon L 
    \end{equation}
\end{lemma}
\begin{proof}
    It can be easily derived since $c_{2n}> c_0 > \frac{L}{2}- \frac{\epsilon L}{4n}$
\end{proof}
\begin{lemma}\label{lemma: c_k bigger x}
    For every $x \in [x_i, y_i]$ (resp $x \in [-y_i, -x_i]$), every $c_k$ from $k = 0,\cdots, 2n$ satisfies $c_k > x$.
\end{lemma}
\begin{proof}
     A point $x$ within the interval $[x_i, y_i]$ is bounded by the distance from $O$ to the point $C$, i.e. the endpoint of the segment. Such distance is $\frac{\sqrt{5}L}{6}$ by the Pythagorean theorem and $c_k> c_0 > \frac{\sqrt{5}L}{6}$ by step $1$. 
     
\end{proof}

\subsection{Main Theorem}\label{subsection: Main Theorem}
The concept behind our construction is to ensure that every solution to the Minsum problem is built by tours such that each tour covers exactly one segment. Moreover, the tours of a solution of the $2$-Min-Makespan problem are exactly the tours of the solution of the Minsum problem distributed optimally. Such optimal distribution is connected with the Minmax Partition Problem as the length of each tour that covers exactly one segment is  $s_i' = Ks_i +C$ by Step 3. The main theorem of this section is derived from these properties, which are expressed here as two statements. For the sake of clarity, we will provide proofs of these statements in the next section.

\par 
\begin{Statement 1}\label{Statement 1}
    Let $\alpha_S$ be the segments in our construction associated with a set of positive rationals $S = \{s_i\}_{i=1}^{2n}$, then the solution of the Minsum problem for one drone is determined by a set of tours $T = \{t_i\}_{i= 1}^{2n}$ with each tour $t_i$ covering exactly the $i$-th segment.
\end{Statement 1}  

\begin{Rem}
Each of the tours $t_i$ of $T$ satisfies $\ell(t_i) = s_i' = Ks_i +C$ by Step 3.
\end{Rem}
\begin{Statement 2}\label{lemma: Statement 2}  The tours of any solution $\{T_1, T_2\}$ of the $2$-Min-Makespan problem for $\alpha_S$ correspond to tours $t_i$ derived from a solution to the Minsum problem.  
\end{Statement 2}
With these two statements, we can now prove:
\begin{theorem}\label{Theorem: NP-Hard}
    The $2$-Min-Makespan problem is Strongly NP-hard.
\end{theorem}

\begin{proof}
   Given a set of positive rationals $S = \{s_i\}_{i=1}^{2n}$ (an instance of Minmax Partition Problem, i.e. Problem \ref{Problem: balanced min max finite partition problem}), obtain $\alpha_S$ in polynomial time using previous construction. Set $K= \frac{L-2c_{2n}}{R}$ and $C = 2c_{2n}$. For $S$, $K$ and $C$, define the set $S'$ as in Proposition \ref{Proposition: KC-proposition}. \par
    Now, given a solution $A$ of Problem \ref{Problem: balanced min max finite partition problem} for $S$, use Proposition \ref{Proposition: KC-proposition} to get a solution $A'$ for $S'$. By Statement $1$, we get a solution of the Minsum problem for $\alpha_S$ with one drone is $T = \{t_i\}_{i= 1}^{2n}$ with each tour $t_i$ covering exactly the $i$-th segment. Thus, assigning each tour $t_i$ to $T_2$ if and only if $s_i' \in A'$, we will have a solution of the $2$-Min-Makespan problem by Statement $2$, the minimality of $A'$ and the fact that $\ell(t_i) = s_i'$.\par
    In the other direction, given a solution $\{T_1, T_2\}$, for the $2$-Min-Makespan problem in $\alpha_S$,  assign $s_i'\in A'$ if and only if $t_i \in T_2$ (assume w.o.l.g that $\ell(T_2)\geq \ell(T_1)$). By Statement $2$, the length of each $t_i$ is $s_i'$; therefore we obtain a solution $A'$ of Problem \ref{Problem: balanced min max finite partition problem} for $S'$ and, by Proposition \ref{Proposition: KC-proposition}, we can consequently derive a solution for $S$.
\end{proof}

\subsection{Proofs of Statement 1 and Statement 2} \label{subsection: Proof of the statements}

From Step 1 we get:

\begin{lemma}\label{lemma: no tour that intersects different segments}
    There is no tour of length at most $L$ that intersects two different segments in  $\alpha_S$. 
\end{lemma}
\begin{proof}
    From the fact that $c_0 >  L-y_0- \frac{L}{3}$, there is a gap between right-hand segments and left-hand segments that can not be covered by a tour of length at most $L$. Therefore there is no tour that intersects a right-hand segment and a left-hand segment. Moreover, as we have taken segments alternatively in each side of $O$ leaving  a hole of a tour of length $L$ between two consecutive ones (by the equality of step $1$: $y_{i+1} -y_{i} +c_{i} +c_{i+1} = L$), there is no tour that connects two consecutive segments. 
\end{proof}

A straightforward consequence of the previous Lemma is {\bf Statement 1}: 

\begin{proof}[{\bf Proof of Statement 1. }]
By Lemma \ref{lemma: no tour that intersects different segments}, if a tour begins in a segment, it must end in the same segment. Furthermore, by triangle inequality, it would take more length to make more than one tour to cover just one segment. So the set of tours $T = \{t_i\}_{i= 1}^{2n}$ with each tour $t_i$ covering exactly the $i$-th segment is optimal for the Minsum problem. 
\end{proof}
\begin{Rem}\label{Remmark: minimum length tour}
    Notice that by Step $1$, every tour $t \in T$ has a length greater than two times $c_0$, so its length is at least $\frac{\sqrt{5}L}{3}$, in particular, its length \textcolor{blue}{is} at least $\frac{2L}{3}$.
    \end{Rem}
From now on, let $T = \{t_i\}_{i= 1}^{2n}$ be the solution of the Minsum problem. To prove {\bf Statement 2}, we need some technical lemmas. We begin with a lemma that shows that if a segment $[x_i, y_i]$ (or $[-y_i, -x_i]$) of our construction is covered by two different drones (\emph{broken segment}), then the slowest drone (the one that makes more length) covers its part only with one tour. 
\begin{lemma}\label{lemma: Only two tours disjoint per segment}
    Let $\{T_1, T_2\}$ be two sets of tours that solve the $2$-Min-Makespan problem for $\alpha_S$ with $\ell(T_2) \ge \ell(T_1)$. For every segment $[x_i, y_i]$ (resp. $[-y_i, -x_i]$ for the left-hand side segments) there is at most a point $r\in [x_i, y_i]$ (resp. $r\in [-y_i, -x_i]$) and at most a tour $t_2$ of $T_2$ such that $t_2$ completely covers the subsegment $[x_i, r]$ (resp.$[r, -x_i]$) or the subsegment $[r, y_i]$ (resp. $[-y_i, r]$).    
\end{lemma}
\begin{proof}
The proof follows easily from the triangle inequality. If there were two (or more) tours $t_2, t_2' \in T_2$ covering different parts of the segment $[x_i, y_i]$ with $t_2$ the tour closest to the segment $OB$, such a solution could be improved by removing $t_2$ and $t_2'$ from $T_2$ and adding the tour $t_2^{''}$ that begins at the beginning point of $t_2$ and ends at the end point of $t_2'$, which would contradict the minimality of $T_2$. 
\end{proof}
    From Lemma \ref{lemma: Only two tours disjoint per segment} we could think that for a broken segment the part not covered by the set of tours
    of $T_2$ might be covered by several different tours of $T_1$. Although from our proof of Statement $2$ it will later be deduced that this is not possible, we can also skip this technicality with the following result: 
\begin{lemma}
    From every optimal solution $\{T_1, T_2\}$ of the $2$-Min-Makespan problem for $\alpha_S$ with $\ell(T_2) \ge \ell(T_1)$ we can find in linear time another solution $\{T_1', T_2\}$  such that every segment $[x_i, y_i]$ (resp. $[-y_i, -x_i]$ for the left hand side segments) intersects $T_1'$ in at most one tour. 
\end{lemma}
\begin{proof}
    Let $[x_i, y_i]$ be a segment that intersects $T_1$ in more than one tour. Assume without loss of generality that the sub-segment $[x_i, r]$ is covered by several tours of $T_1$.
    From the proof of the previous Lemma \ref{lemma: Only two tours disjoint per segment}, it is clear that we could merge in linear time such tours of $T_1$ to get only one tour that covers $[x_i, r]$, decreasing the total length of the new $T_1'$. \par
\end{proof}
From the previous Lemmas, we conclude:

\begin{lemma}\label{corollary: Broken segments are made only with two tours}
From every optimal solution $\{T_1, T_2\}$ of the $2$-Min-Makespan problem for $\alpha_S$ with $\ell(T_2) \ge \ell(T_1)$ we can find in linear time another optimal solution $\{T_1', T_2\}$  such that every segment $[x_i, y_i]$ is covered by at most two tours one from $T_1'$ and another from $T_2$ and such tours intersects in at most one point $r \in [x_i, y_i]$. Analogously for segments $[-y_i, -x_i]$. 
    
\end{lemma}
The following result will show that in the optimal solution of the $2$-Min-Makespan problem for $\alpha_S$, there is only at most one broken segment.
\begin{lemma} \label{lemma: Only one broken segment}
    Let $\{T_1, T_2\}$ be two sets of tours that solve the $2$-Min-Makespan problem for $\alpha_S$. Then there is at most one broken segment.   
\end{lemma}
\begin{proof}
The idea of the proof follows from the fact that if there were two segments then the solution could be improved by merging tours that cover the same segment, as can be easily checked by comparing the two distributions of Figure  \ref{fig: Two broken segments}. We will just do the calculations. 
\begin{figure}[ht]
   \centering
   \begin{tabular}{c c c}
 \includegraphics[width=.45\textwidth,page=3]{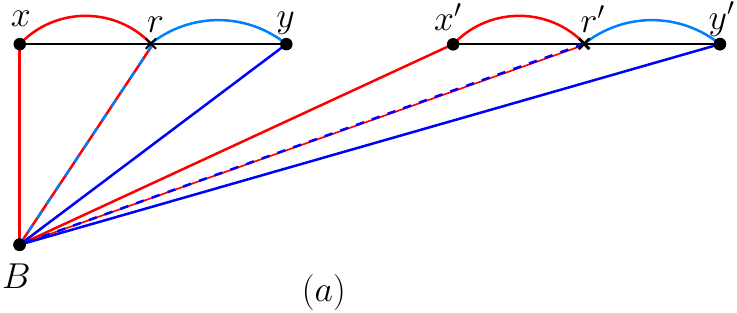} 
  & &
 \includegraphics[width=.45\textwidth,page=4]{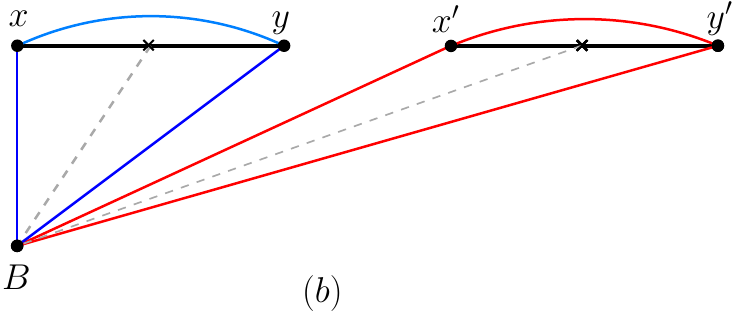} 
    \end{tabular}
    \caption{(a) Two broken segments.  (b) Reassigning the tours, the global length is improved.}
    \label{fig: Two broken segments}
\end{figure}

Assume that there is a pair of broken segments $[x, y]$ and $[x',y']$. By Lemma \ref{corollary: Broken segments are made only with two tours}  each one is covered with two tours: one tour from $T_1$ and another tour from $T_2$ (without loss of generality we assume that both segments are on the same side). By Lemma \ref{lemma: Only two tours disjoint per segment} there is a point $r \in [x, y]$ and another one $r' \in [x', y']$ such that the tours\footnote{In order to make easier the reading of this proof we abuse of notation and define the tours by the segments not by the vertices, for example, the tour $t_1$ should be defined by $t = BxrB$ with $b = d(B, x)$, $r-x = d(r, x)$ and $d = d(r, B)$, but we define it by $t_1 = b(r-x)d$.}  $t_1  = b(r-x)d$ and $t_2 = d(y-r)c$ belong to different drones with $b = d(B, x)$, $d = d(B,r)$ and $c = d(B, y)$ (resp. $t_1' = b'(r'-x')d'$ and $t_2' = d'(y'-r')c'$). We will do the case $t_1, t_1' \in T_1$ and $t_2, t_2' \in T_2$. The other case is analogous and it is left to the reader.

Let $\ell(T_1) = M_1 + b + r-x + d + b' +r'-x' +d'$  and $\ell(T_2) = M_2+ d + y-r+c + d'+y'-r'+c'$  for some $M_1$ and $M_2$ non negative numbers. Assume without loss of generality that $\ell(T_2) \ge \ell(T_1)$, and assume also that $b +y-x + c \ge b' +y'-x' + c'$. Then we build $T_2'$ by removing from $T_2$ the tours $t_2$ and $t_2'$ and adding the tour $t_2'' = b(y-x)c$  and $T_1'$ by removing from $T_1$ the tours $t_1$ and $t_1'$ and adding the tour $t_1'' = b'(y'-x')c'$. We get that the new lengths are $\ell(T_1') = M_1 + b'+ y'-x'+c'$ and $\ell(T_2') = M_2 +b'+ y-x+c$. But this would imply: 
\begin{equation}
        \notag
        \ell(T_1)-\ell(T_1')  =  b +r-x +d +r' +d'-y'-c'>0
    \end{equation}
   Last inequality holds as $d' > y'$ by Lemma \ref{lemma: c_k bigger x} and $b + d > c'$ as $c'< L/2$ and $b+d > 2L/3$.
Also, we get: 
\begin{equation}
        \notag
        \ell(T_2)-\ell(T_2')  = d-r+d'+y'-r'+c'- b +x>0
    \end{equation}
Again we get the last inequality because $d > b$ and $d'> r$ by Lemma \ref{lemma: c_k bigger x}. Therefore $\ell(T_1) > \ell(T_1')$ and $\ell(T_2) > \ell(T_2')$ what contradicts the minimality of $\{T_1, T_2\}$.  
\end{proof}

We also require the following technical Lemma to prove Statement 2:

\begin{lemma}\label{lemma: subset of the same cardinality less L/3}
Let $T = \{t_i\}_{i= 1}^{2n}$ be the solution of Minsum problem for $\alpha_S$.    Let $T_1'$ and $T_2'$ be two (disjoint) subsets of $T$ of the same cardinality, then $|\ell(T_2')-\ell(T_1')| < \epsilon L$, where $\epsilon$ was defined in the Step 1  of the construction of $\alpha_S$.
\end{lemma}

\begin{proof}
    Assume without loss of generality that $\ell(T_2') \ge \ell(T_1')$ and let $k \le  n$ be the cardinality of $T_1'$ and $T_2'$. As $s_i' \in (2c_{2n}, L]$. Then $\ell(T_2') = \sum_{i\in T_2'} s_i' \le k\cdot L$ and $\ell(T_1') = \sum_{i\in T_1'} s_i' \ge k \cdot 2c_{2n}$. Therefore: 
    \begin{equation}
        \notag
        \ell(T_2')-\ell(T_1') \le k(L- 2c_{2n}) \le 2n(L-2c_{2n}) < \epsilon L
    \end{equation}
    The last inequality follows from Lemma \ref{lemma: Condition of cardinality}.
\end{proof}
Finally, we get {\bf Statement $2$}:
\begin{proof}[{\bf Proof of Statement 2.}] 
From Lemma \ref{lemma: Only one broken segment}, we have that $T_1$ and $T_2$ are composed of tours from $T$, except for just two of them defined by a broken segment. We are going to show that it is impossible that such a broken segment exists. Assume, for the sake of contradiction, that such an element exists and let $t_1 \in T_1$ and $t_2 \in T_2$ be the tours that cover the subsegments of the broken segment. Assume also that $\ell(T_2) \ge \ell(T_1)$.  Let $T_2' = T_2 \setminus t_2$ and $T_1' = T_1 \setminus t_1$. Notice that $\ell(T_2) = \ell(T_2') + \ell(t_2)$, $\ell(T_1) = \ell(T_1') + \ell(t_1)$ and  $ \frac{2L}{3}< \ell(t_i) \le L$ for $i =1, 2$.  
    \par
First, we claim that the cardinalities satisfy  $|T_1'| \ge |T_2'|$, otherwise, if $|T_1'| < |T_2'|$, there would be a proper subset of tours $T_2'' \subset T_2'$ such that $|T_2''| =|T_1'|$, then  by Lemma \ref{lemma: subset of the same cardinality less L/3}, we would get that $|\ell(T_1') - \ell(T_2'')| < \epsilon L$. By Remark \ref{Remmark: minimum length tour}, each tour  $t \in T_2' \setminus T_2''$ would add at least $\frac{2 L}{3}$ to the total length of $T_2'$, then:
    \begin{equation*}
    \notag
        \ell(T_2') > \ell(T_1') + \frac{2 L}{3}- \epsilon L
    \end{equation*}    
     \begin{equation*}
     \notag
         \ell(T_2) -\ell(T_2') > \frac{2L}{3}
     \end{equation*}
Now we can remove $t_2$ from $T_2$ and merge it with $t_1$ in $T_1$ to get just one tour $t$  and we will get that the new solution would be the maximum of $\ell(T_2')$ and $\ell(T_1') + \ell(t)$. If $\ell(T_2')$ is the maximum, we have improved $\ell(T_2)$, which is a contradiction with the minimality of $T_2$. Now, if $\ell(T_1')+\ell(t)$ is the maximum, it would exceed $\ell(T_2')$ by at most:
    \begin{equation}
        \notag
        \ell(T_1')+\ell(t) - \ell(T_2') < \ell(t) + \epsilon L-\dfrac{2 L}{3} \le (1+\epsilon)L-\dfrac{2 L}{3} < \dfrac{2L}{3} < \ell(T_2)- \ell(T_2')
    \end{equation} 
    That implies that $\ell(T_1')+\ell(t) < \ell(T_2)$ and again, it contradicts the minimality of $T_2$.
    \par 
    So we have proved that $|T_1'| \ge |T_2'|$, but as the number of segments is even and all the tours of $T_1'$ and $T_2'$ are all tours of the one drone solution except one (the broken segment), it can not happen $|T_1'| = |T_2'|$. We get $|T_1'| > |T_2'|$. Again, let $T_1''$ be a subset of tours of $T_1'$ of the same cardinality as $T_2'$. Using analogous reasoning as in the previous paragraph we get:
    \begin{equation*}
        \notag
        \ell(T_1')> \ell(T_2')+\frac{2 L}{3}- \epsilon L
    \end{equation*}
    what would imply 
    \begin{equation*}
      \notag
        \ell(T_1) > \ell(T_1') +\dfrac{2L}{3} >\ell(T_2') + \frac{4 L}{3}- \epsilon L
    \end{equation*} 
    and then, as  $\ell(T_2') + L \ge \ell(T_2)$, we have: 
    \begin{equation*}
        \notag
        \ell(T_1) >  \ell(T_2) + \frac{L}{3}- \epsilon L > \ell(T_2)
    \end{equation*} 
    But we have assumed $\ell(T_2) \ge \ell(T_1)$, which leads to a contradiction.
\end{proof}

\section{Approximation algorithms}\label{section: Aproximation algorithm}

In this section, we present a simple yet competitive approximation algorithm, followed by an improved version of it. As a starting set of tours, we use the optimal solution of the minsum problem proposed in \cite{bereg2024covering}:

Given a set of segments on a straight line and an upper bound $L>0$ for the tour lengths, the goal of the minsum problem is to compute a set of tours to cover all segments so that the total length is minimized. The authors show that the problem can be solved in polynomial time using a dynamic programming approach.

\subparagraph{Greedy tour distribution for two drones (\textsc{G2D-algorithm}):}
Let $T = \{t_1, \cdots, t_m\}$ be an optimal set of tours for the minsum problem.
The \textsc{G2D-}algorithm simply distributes the tours of $T$ into two sets, $T_1$ and $T_2$, so that, in each step, the difference between the sum of the lengths of the tours in each set is $| \ell(T_2)-\ell(T_1)|  = aL$ {\color{red}for some}
$0\le a\le 1$. To do it, given $T = \{t_1, \cdots, t_m\}$, add for $i= 1$ the tour $t_1$ to $T_1$ and, in each subsequent step $i>1$, add the tour $t_i \in T$ to the set $T_1$ or $T_2$ with minimum total length.

\begin{Rem}\label{obs1}
    Assume, without lost of generality, that $\ell(T_2) \ge \ell(T_1)$ in G2D-algorithm; then, for some $a \in [0, 1]$,  $\ell(T_2) = \ell(T_1) +aL$. 
\end{Rem}
The interesting fact about this simple greedy algorithm is that it is output-sensitive and it behaves nicely in practical situations. The following result shows the approximation factor:
\begin{theorem} \label{Theorem: Aproximation Theorem}Let $\{T^*_1, T^*_2\}$ be any  solution of the $2$-Min-Makespan  problem and let $\{T_1, T_2\}$  be the final distribution of the \textsc{G2D-}algorithm. Assume wlog that $\ell(T_1) \le \ell(T_2)$ and $\ell(T^*_1) \le \ell(T^*_2)$. Then:
\begin{itemize}
    \item[a)]  $\ell(T_1)+\dfrac{aL}{2} \le \ell(T^*_2) \le \ell(T_1) +aL $, where $a \in [0, 1]$.
    \item[b)] If $a = 0$, then $\{T_1, T_2\}$ is optimal for the $2$-Min-Makespan problem.
    \item [c)] If $a \in (0, 1]$, then $\ell(T_2) \le \Delta \cdot \ell(T_2^*)$, with $\Delta = \dfrac{\Gamma+2}{\Gamma+1}$ and $\Gamma = \dfrac{ 2\ell(T_1)}{aL}$. 
\end{itemize}
\end{theorem}

\begin{proof}
\begin{itemize}
   \item[ \emph{a)}] If $\ell(T^*_2) < \ell(T_1)+\dfrac{aL}{2}$, using Remark \ref{obs1} we have: 
\begin{eqnarray*} \ell(T_2^*)+\ell(T_1^*)< 2 \ell(T_1) +aL = \ell(T_2) + \ell(T_1) = \ell(T)
    \end{eqnarray*}
which contradicts the minimality of $T$.
    \par
\item[ \emph{b)}]  
Follows trivially from a). 
\par
\item[ \emph{c)}] 
Let $\Delta= \dfrac{\ell(T_1)+aL}{\ell(T_1) + \frac{aL}{2}}$. From a) we get: 
\begin{equation*}
\notag
        \ell(T_2) = \ell(T_1) + aL = \Delta\cdot \left(\ell(T_1) + \dfrac{aL}{2}\right)\le \Delta \cdot \ell(T_2^*).
\end{equation*}
\par
Therefore, making $\Gamma =\dfrac{ 2\ell(T_1)}{\ell(T_2)-\ell(T_1)} = \dfrac{2\ell(T_1)}{aL}$, we obtain c).
\end{itemize}
\end{proof}

\begin{Rem}
The computational time of \textsc{G2D-}algorithm is essentially given by the computation of the set $T$. Thus, according to \cite{bereg2024covering}, \textsc{G2D-}algorithm computes $T_1$ and $T_2$ in $O(n^2+nm)$ time,  where $n$ is the number of segments and $m$ the total number of tours.  
\end{Rem}
\begin{Rem}
Notice that when $\ell(T_1)$ increases then $\Delta$ goes to $1$. This situation can arise when there are many segments and $L$ is small.
In particular, note that if $\ell(T_1) \ge 3L$, then $\Delta \le 1.14$. Therefore, this approximation algorithm would perform reasonably well in practice, especially in scenarios where multiple tours are required. 
\end{Rem}

{\color{violet}HIGES}

\begin{Rem}\textsc{G2D-}algorithm  can be generalized for $k>2$ drones while maintaining the same time complexity; we just have to properly distribute the $m$ tours of $T$ into $k$ sets $T_1, \cdots, T_k$. It is easy to see that the approximation factor is $\Delta = \dfrac{\Gamma+k}{\Gamma+1}$ with $\Gamma = \dfrac{k\cdot \ell(T_1)}{\ell(T_k)-\ell(T_1)}$, where $T_1$ is the set with minimum total length and $T_k$ the set with maximum total length.
\end{Rem}

\subsection{An improved approximation}
After applying the G2D-algorithm to $n$ segments, we propose two strategies for improvement.
\begin{itemize}
    \item {\bf Cutting Strategy (CS):} Let $\{T_1, T_2\}$ be the set of tours obtained by the G2D-algorithm. Assume wlog that $\ell(T_2) > \ell(T_1)$. Let $t \in T_2$ be a tour that covers $[x, y]$, $t= BxyB$ (notice that $[x, y]$ is not necessarily one of the original segments). If $2\cdot \min\{d(B, x), d(B, y)\} < \ell(T_2)- \ell(T_1)$, 
    we can ``\emph{cut}'' $t$ by computing a point $r \in [x, y]$  such that the tours  $t' = BxrB$ and $t'' = BryB$ can be added, one to $T_1$ and the other to $T_2$, thereby improving the G2D-solution· For example  $T_1' = T_1 \cup t''$  and $T_2' = (T_2 \setminus t) \cup t'$ and $\ell(T_2') \ge \ell(T_1')$. See Figure \ref{fig: cutting and enlarging}(a) for an example. 
    \item {\bf Enlarging Strategy (ES):} Let $\{T_1, T_2\}$ be the set of tours obtained by applying G2D-algorithm. Assume wlog that $\ell(T_2) > \ell(T_1)$. Let $t_1\in T_1$ and $t_2 \in T_2$ be two tours such that $t_1 = BxyB$,  $t_2 = ByzB$, and $\ell(t_1) < L$. If such pair of tours exists, we can ``\emph{enlarge}'' the tour $t_1$ and reduce the tour $t_2$, improving the G2D-solution. To achieve this, compute an $r \in [y, z]$ such that the tour $t_1' = BxrB$ satisfies $\ell(t_1')\le L$. Then, set $T_1' = (T_1 \setminus t_1) \cup t_1'$ and $T_2' = (T_2 \setminus t_2) \cup t_2'$, where $t_2' = BryB$, ensuring that $\ell(T_1')\le \ell(T_2') < \ell(T_2)$. See Figure \ref{fig: cutting and enlarging} (b) to illustrate the strategy.     
\end{itemize}
\begin{figure}[ht]
\begin{tabular}{cc}
\hspace{-0.5cm}
\includegraphics[width=0.5\textwidth, page=7]{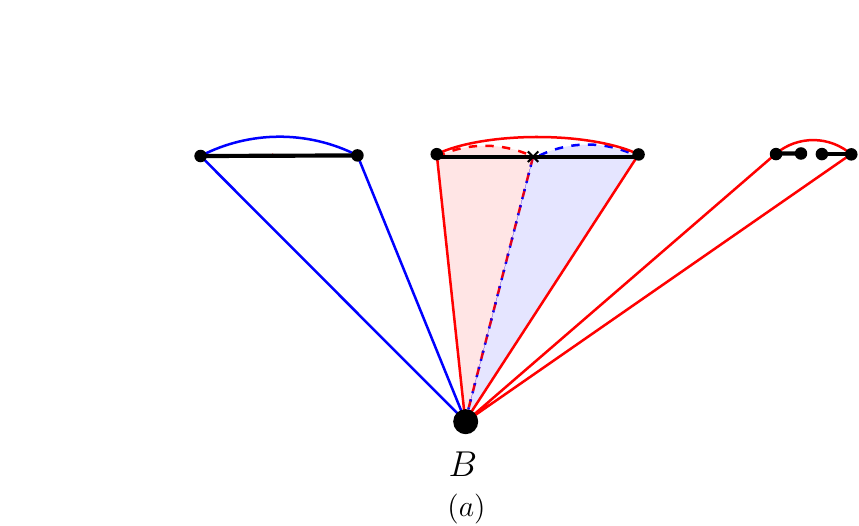}& \hspace{-1cm}
\includegraphics[width=0.5\textwidth, page=9]{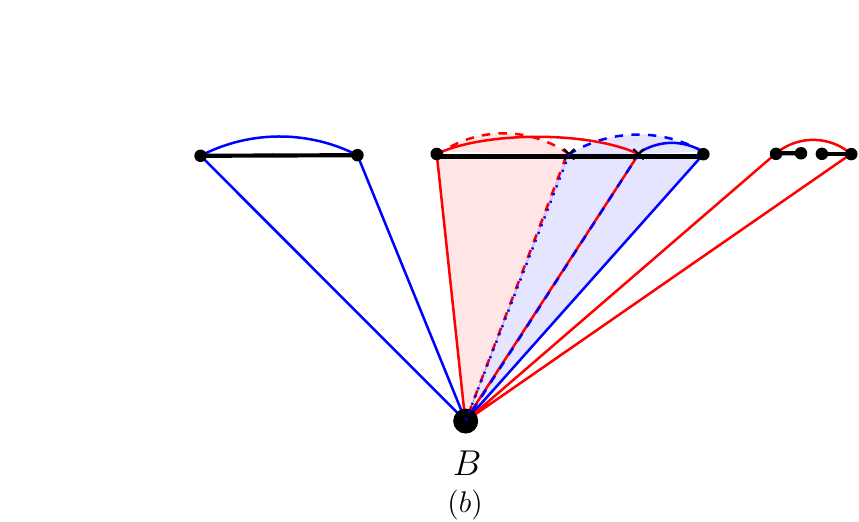} \end{tabular}
\caption{(a) Cutting a red tour. (b)  Enlarging a blue tour and reducing a red tour.}
\label{fig: cutting and enlarging}
\end{figure}
With these two strategies we proceed as follows: First, we apply the G2D-algorithm to obtain an initial solution. Next, we evaluate all possible candidates for cutting and enlarging, which can be computed in linear time. Finally, we implement the strategy that yields the most significant improvement in the maximum total length. If we get the equality in the total length of the two sets of tours, we stop the algorithm; otherwise, we check the remaining candidates for enlarging or cutting and take again the best improvement. In each step, either we stop or the number of candidates for enlarging is reduced by one. It is important to note that the cutting strategy ensures equality in the total length of the set of tours, as the maximum difference between them is at most $L$.
This procedure takes $O(n^2)$ time, leading to an overall time complexity for this improvement of $O(n^2+nm)$ where $n$ is the number of segments and $m$ the number of obtained tours.

\section{Computational experiments}\label{sec:comp}
We aim to evaluate the effectiveness of the G2D-algorithm 
and the improvement
with the cutting and enlarging strategies. For the sake of comparison, we will consider a discrete version of the problem in which the segments are subdivided into smaller subintervals. The tours can partition the intervals using the endpoints of these subintervals.   We will solve this digitalized version using a Mixed-Integer Linear Programming (MILP) formulation. 
Note that a solution for the discrete version simulates an optimal solution for the 2-Min-Makespan problem when the intervals are sufficiently small.
\par

Let $\alpha$ be a set of segments. Take $\delta >0$  as the level of discretization, and assume, without loss of generality, that for every segment $[x_i, y_i] \in \alpha$, the number $\frac{y_i-x_i}{\delta}$ is a positive integer. Now, we will identify the main variables and parameters for the MILP formulation. For the sake of generality, the formulation for the $k$-Min-Makespan problem is presented:
 
\begin{itemize}
    \item Let $\alpha$ be a set of segments. We define a discrete sequence of points $A=\{a_i\}_{i= 1}^n$ as follows\footnote{We will abuse of notation by considering $a_i$ simultaneously a point in the horizontal line of the form $(a_i, 0)$ and its first coordinate.}. Let $a_1$ be the left vertex of the segment that is furthest to the left. Given $a_i$, if $a_i$ is not a right vertex of a segment of $\alpha$, then set $a_{i+1} = a_i+\delta$; otherwise, $a_i$ is a right vertex and let $a_{i+1}$ be the left vertex of the segment to the right of $a_i$.
    \item $B$ will be the base point and $L>0$ the maximum length of a tour.
    \item $\beta_i$ will be the distance from the base $B$ to the point $a_i$.
\end{itemize}
Assume that we have $k$ drones. We make the following restriction:  every tour has to begin in an $a_i$ and end in an $a_j$ with $j>i$. Consider a feasible solution  $\{T_1, \cdots, T_k\}$ of the $k$-min-makespan problem. We define the following variables and parameters:  
\begin{itemize}
    \item For each $i, j \in \{1, \cdots , n\}$ and for each $d\in \{1, \cdots, k\}$, we define the binary decision variables $Z=\{z_{ij}^d\}$, such that $z_{ij}^d = 1$ if and only if there is a tour $t \in T_d$ that begins in $a_i$ and ends in $a_j$ 
    \item For each $r\in \{1, \cdots, n\}$ such that $a_r$ is not a right vertex of a segment, we define the binary decision variable $s_{r(r+1)}^d$ such that $s_{r(r+1)}^d = 1$ if and only if there is a tour $t \in T_d$ such that the segment $[a_r, a_{r+1}] \subset t$.
    \item For simplicity, we will denote by $c_{ij}$ the length of a tour that begins at $a_i$ and ends at $a_j$, $c_{ij} = (\beta_i+\beta_j+a_j-a_i)$. 
    \item Let $T$ be a continuous variable that we aim to minimize in the formulation.
\end{itemize}
With these notations, the $k$-Min-Makespan problem can be formulated as follows: 
\begin{align}
\label{eq:f.0} \mbox{minimize}~& \;T \\
\label{eq:f.1} \mbox{s.t. } & \sum_{i,j=1}^n z_{ij}^d c_{ij} \leq T & \forall d\in \{1, \cdots, k\}\\
\label{eq:f.2}& z_{ij}^d c_{ij}  \leq L & \forall d\in \{1, \cdots, k\}, \forall i, j \in \{1, \cdots, n\}\\
\label{eq:f.3}& s_{q(q+1)}^d = \sum_{i\leq q < j}^n z_{ij}^d & \forall d\in \{1, \cdots, k\}, \forall i, j \in \{1, \cdots, n\}\\
\label{eq:f.4}&\sum_{d=1}^m s_{r(r+1)}^d=1 & \forall r \in {1, \cdots, n} \text{ such that } r \text{ is not an ending point.}
\end{align}

Equation \ref{eq:f.1} guarantees that $T$ will be the maximum length among all the lengths of the tours $\{T_1, \cdots, T_m\}$ for an optimal solution. 
Equation \ref{eq:f.2} limits the length of each tour by $L$. 
Equations \ref{eq:f.3} and \ref{eq:f.4} are needed to guarantee that all segments are covered and that each subinterval $[a_r, a_{r+1}]$ is covered exactly with one drone, respectively.

\par
In the follows, we present a systematic performance comparison between optimal solutions obtained via MILP formulation (Gurobi™ solver) on the discretized problem and the heuristic solutions generated by the G2D-algorithm in both its baseline implementation and the improved version with cutting/enlarging strategies.
We will call this improvement the \emph{ improved algorithm}. The idea is to test the behavior of our algorithms under different practical conditions. Such conditions will appear after the modification of different variables. In particular, we will study three variables: the \emph{maximum length} $L$ of a tour, the \emph{homogeneity} in segment length distribution, and the segment \emph{density}. These last two variables are of practical interest in some application fields, as for example, in power plant inspections. We will elaborate on the details of each variable in the description of each set of experiments. 

In the experiments, we randomly generate a set of segments $\{[x_i, y_i]\}_{i=1}^{N}$ within an interval $[m, M]$. The discretization level $d$ is defined as the maximum possible distance between the left endpoint $x_i$ of the leftmost segment and the right endpoint $y_j$ of the rightmost segment.
The base is positioned at $(X_b, Y_b)$, with all segments lying along the line $Y= 0$.
The drone’s maximum length $L$ is randomly selected from $[L_{min}, L_{max}]$, where:

\begin{itemize}
    \item $L_{min} = 2\max\{d(B, a_1), d(B, a_n)\}$ (the minimum length required to reach the farthest segment and return).
    \item $L_{max} = d(B, a_1) + a_n - a_1 + d(B, a_n)$ (the minimum length needed to cover the entire line in one trip).
\end{itemize}

\subsection{Experiment 1}
In the first experiment, we investigate how the segment homogeneity and the maximum tour length $L$ influence the approximation factor.
We generated a set of $987$ scenarios in which $m = -1000$, $M = 1000$ and $d = 1000$. 
Note that while individual segments may lie anywhere within $[-1000,1000]$, the maximum distance between the leftmost and rightmost endpoints of any two segments is constrained to 1000.
$L$ was randomly generated between $[L_{min}, L_{max}]$, $X_b$ between $[0, 1000]$ and $Y_b$ between $[1, 10000]$. 
Segment homogeneity is measured as the coefficient of variation $CV$ of the length of the segments. Three levels of $CV$ are considered: high homogeneity $CV = 0.2$ ($429$ cases), medium homogeneity $CV = 0.6$ ($278$ cases), low homogeneity $CV = 0.8$ ($280$ cases). 
The segment length distribution was designed with a fixed mean of $10$.

In this experiment, we evaluate two metrics: the approximation factor of the algorithm $G2D$, $\Delta$, and the approximation factor of the improved algorithm, $\Delta_I$. If $T_2$ is the solution of one of our algorithms and $T_2*$ is the solution obtained from the solver, the approximation factor is measured as the quotient $\ell(T_2)/\ell(T_2*)$. \par
In Figure \ref{fig:apr1}a) it is shown the distribution of $\Delta$ for the G2D algorithm. The optimal $\Delta = 1$ is achieved in $10.84\%$ of cases with a mean of $1.094$ and a standard deviation of $0.11$. Only $5\%$ of the cases exceed $1.314$ with a maximum of $1.831$. 
Figure \ref{fig:apr1}b) illustrates the distribution of $\Delta_I$ for the improved algorithm. The optimal $\Delta = 1$ is achieved in $16.31\%$ of the cases with a mean of $1.04$ and a standard deviation of $0.045$. Only $5\%$ of the cases exceed $1.1288$ with a maximum of $1.195$. The mean $\Delta_I$ of the improved  algorithm is significantly lower than the mean $\Delta$ of the G2D-algorithm ($t-Student = 13.93$, $p< 0.001$). 

\par
We conducted a one-way ANOVA to quantify the effect of segment homogeneity on algorithm performance metrics.
The homogeneity does not affect significantly neither $\Delta$ ($F = 0.458$, $p > 0.1$) 
nor $\Delta_I$ ($F = 0.146$, $p > 0.1$). \par

\begin{figure}[ht]
\begin{tabular}{cc}

\includegraphics[width=0.45\textwidth, page=1]{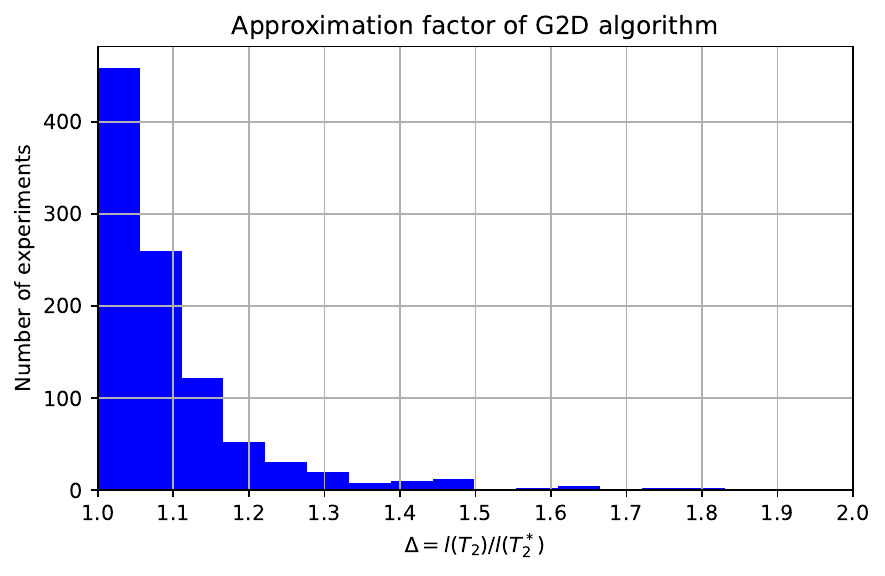}& \hspace{0.2cm}
\includegraphics[width=0.45\textwidth, page=1]{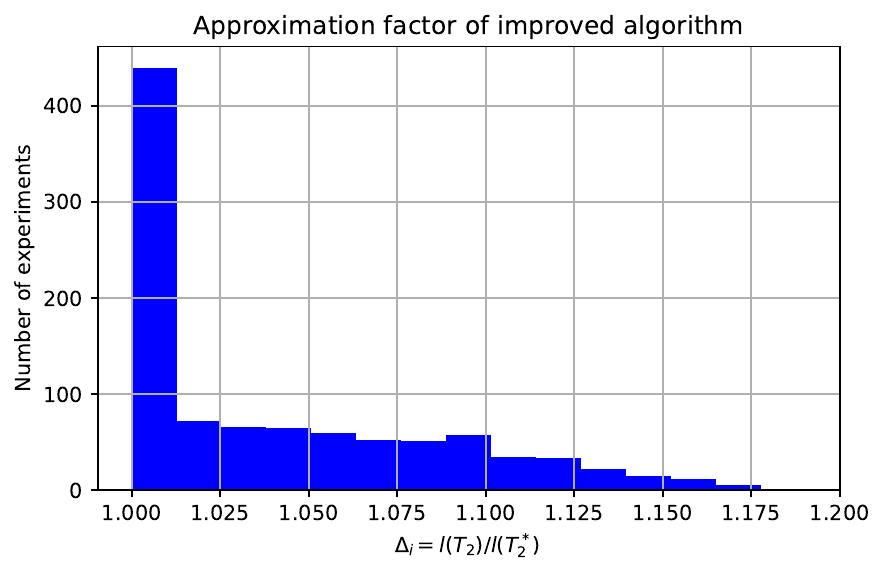} 
 \\
(a) & (b)
\end{tabular}
\caption{a) Approximation factor $\Delta$, where $\ell(T_2)$ is a solution received from G2D-algorithm and $\ell(T_2^*)$ is an optimal solution found by Gurobi™.
(b) Approximation factor $\Delta_I$, where $\ell(T_2)$ is a solution received after applying the Cutting and Enlarging strategies.}
\label{fig:apr1}
\end{figure}

To analyze how $\Delta$ and $\Delta_I$ vary with $L$, we rely on Theorem \ref{Theorem: Aproximation Theorem}(c), which predicts that $\Delta$ grows monotonically with $L$. 
There is no expected pattern for $\Delta_I$. 
$L$ is normalized by defining a scaled parameter:
 $L' = \frac{L}{L_{max}}$. 
Notice that $L'$ is bounded by $1$ and can be close to $0$ depending on $Y_b$ and the configuration of the segments. This modification is required as the position of the base varies randomly. With abuse of notation, from now, we refer to this new $L'$ as $L$. As expected, Figure \ref{fig:apryL}a) shows that when the maximum length of a tour $L$ increases, the aproximation factor of the G2D-algorithm increases.  
Figure \ref{fig:apryL}b) reveals no monotonic relationship between the improved algorithm's approximation factor $\Delta_I$ and tour length $L$, contrasting with the baseline's clear trend.

\par

\begin{figure}[ht]
   \centering
\includegraphics[width=0.7\textwidth]{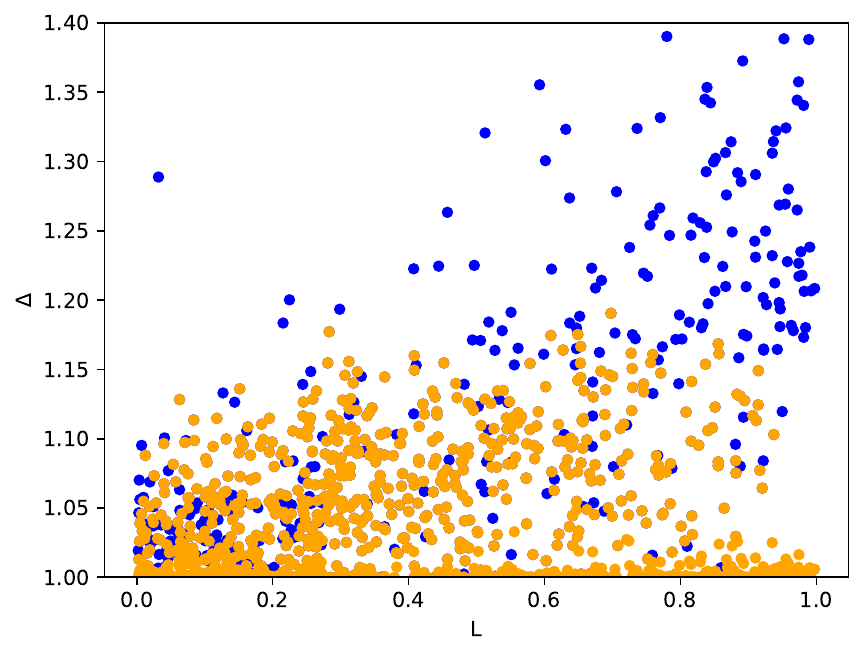}    \caption{Changes in the approximation factor of the G2D-algorithm when the maximum length of a tour L increases (in blue), and changes of the improved algorithm (in orange).}
    \label{fig:apryL}
\end{figure}

To further investigate the impact of tour length, we operationalized $L$ by dichotomizing it into two distinct experimental levels:
low $L$ (if $L < 0.5$) and high $L$ (if $L \ge 0.5$).
We now investigate how the dichotomized drone length $L$ affects both approximation factors $\Delta$ and $\Delta_I$, while also examining potential interaction effects with segment homogeneity.
We perform two $3 \times 2$ ANOVA test, one on $\Delta$ and another on $\Delta_I$ with at least $50$ experiments per group. In both cases, $\Delta$ and $\Delta_I$, the dichotomized factor $L$ significantly affected the approximation factor ($F = 232.110$, $p< 0.001$ for $\Delta$ and $F = 25.272$, $p< 0.001$ for $\Delta_I$).  Table \ref{Table: Llowhigh} shows the means of the approximation factors for low and high $L$. For both algorithms (greedy and improved), we observe a monotonic relationship between drone length $L$ and approximation performance.

 \begin{table}[ht]
\centering
\begin{tabular}{ | m{4cm} | m{4cm}| m{4cm} | } 
  \hline
 Approximation factor & low $L$ \par $N = 485$  & high $L$ \par $N = 502$\\ 
  \hline
  $\Delta$ &  $\bar{\Delta} =1.044 $ \par $\sigma_{\Delta} =  0.041$  &  $\bar{\Delta} =1.142 $ \par $\sigma_{\Delta}=0.137 $  \\ 
  \hline
  $\Delta_I$ &  $\bar{\Delta_I} =1.033  $ \par $\sigma_{\Delta_I} =  0.038$  &  $\bar{\Delta_I} =1.047$ \par $\sigma_{\Delta_I}= 0.05$ \\ 
  \hline
\end{tabular}
\caption{Means and standard deviations of approximation factor with low and high $L$}
 \label{Table: Llowhigh}
\end{table}

Finally, we found an interaction effect between the dichotomized $L$ and the homogeneity variable for $\Delta_I$ ($F = 12.971$, $p< 0.001$) but not for $\Delta$ ($F = 1.848$, $p > 0.1$). Figure \ref{fig:figInterLHomo} shows the interaction pattern of both variables.  

\begin{figure}[ht]
   \centering
\includegraphics[width=0.7\textwidth]{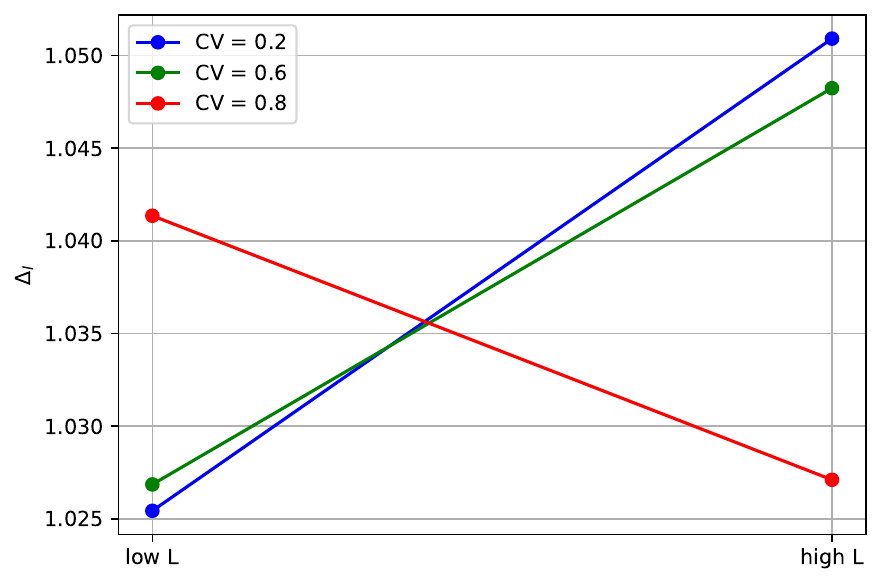}
    \caption{Interaction effect between $L$ and homogeneity factor for 
    $\Delta_I$.}
    \label{fig:figInterLHomo}
\end{figure}

\subsection{Experiment 2}
In this experiment, we include the variable density in the study.
Let $d$ be the level of discretization. Recall that it is defined as the maximum possible distance between the left endpoint $x_i$ of the leftmost segment and the right endpoint $y_j$ of the rightmost segment.  Let $S_1$ be the number of subintervals of length one contained within the union of all randomly generated segments $[x_i, y_i]$. Then the density $\rho$ is defined as $\rho = \frac{S_1}{d}$.

We generate a set of $3306$ scenarios with $m = -250$, $M = 250$ and $d = 500$. 
 In this experiment, we reduced the discretization level $d$ to compensate for the computational limitations of the Gurobi™ solver, as evidenced by the execution times reported later.
 To systematically evaluate the effect of the length $L$, we fixed the base station at coordinates $(X_b,Y_b) = (250,500)$ and randomly assigned $L$ from three predetermined levels: low $L$ with $L \in [L_{min}, L_{min} +15]$ ($1076$ cases), medium $L$ with $L \in [L_{min} +15, L_{min} +120]$ ($1143$ cases), and high $L$ with $L \in [L_{min} +120, L_{max}]$ ($1087$ cases). These ranges were selected through preliminary testing, as they were found to generate sufficient variation in the outputs.
As a third experimental factor, we reintroduced segment length homogeneity, measured identically to previous experiments but now discretized into two levels (high vs. medium/low). This simplification was justified by the similar behavioral patterns observed between medium and low homogeneity conditions in prior results (Figure \ref{fig:figInterLHomo}).
High homogeneity was set as $CV = 0.2$ ($1675$ cases) and low homogeneity as $CV = 0.8$ ($1631$ cases). 
 
In each scenario, the segments were generated randomly with normal distribution, with mean $m$ uniformly chosen with $m \in [10, 100]$. 
The density factor was divided into two levels: low density $\rho \approx 0.2$ ($2526$ cases)   and high density $\rho \approx 0.8$ ($780$ cases). 
Our segment generation algorithm iteratively adds segments until achieving approximately the required density. There are much lower cases of high density rather than low density due to problems in the execution times of the solver.
Then, we considered a $3 \times 2 \times 2$ factorial design in which the number of cases in each group was at least $120$. In addition, we studied the execution time, $T_G$, of the Gurobi™ solver. The maximum execution time of our algorithms, both the G2D-algorithm and the improved version, was less than $0.1$ seconds. However,
$T_G$ has a mean of $69.174$ seconds, $\sigma = 523.64$ and a maximum of $22438.413$ seconds. \par

As a first result, we found a similar pattern as in the previous experiment when comparing $\Delta$ and $\Delta_I$ as can be seen in Figure \ref{fig:apryL2}. Table \ref{Table: DeltaDeltai} shows the mean, the standard deviation, the maximum and the percentage of cases with an approximation factor equal to $1$ of both variables. As expected, the mean of $\Delta_I$ was significantly lower than the mean of $\Delta$ ($t-Student = 28.601$, $p< 0.001$). 

\begin{figure}[ht]
   \centering
\includegraphics[width=0.7\textwidth]{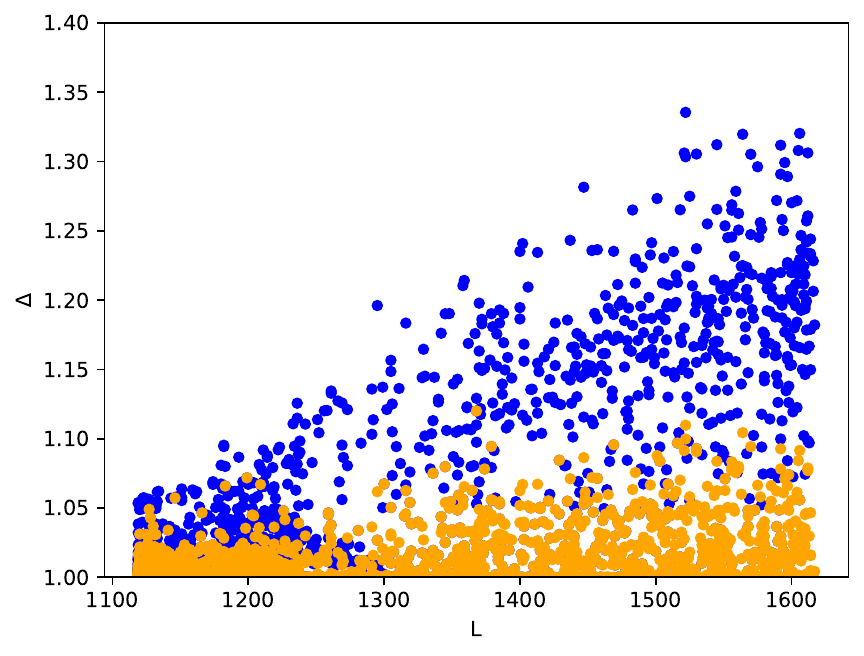}    \caption{Impact of maximum tour length $L$ on approximation factors for (1) the baseline G2D-algorithm (blue) and (2) the improved algorithm (orange).
}
    \label{fig:apryL2}
\end{figure}

 \begin{table}[ht]
\centering
\begin{tabular}{ | m{3cm} | m{2cm}| m{2cm} | m{2cm}| m{2cm}|} 
  \hline
 Approximation factor & Mean & $\sigma$ & Maximum & Exact \\ 
  \hline
  $\Delta$ &  $1.033$ &   $0.061$ & $1.335$ & $36.33\%$ \\ 
  \hline
  $\Delta_I$ &  $1.0072$ & $0.015$ &  $1.12$ & $63.25\%$\\ 
  \hline
\end{tabular}
\caption{Approximation factor results in the Experiment 2.}
\label{Table: DeltaDeltai}
\end{table}

Now, to analyze the effects of $L$, segment homogeneity and segment density on performance metrics $\Delta$, $\Delta_I$ and $T_G$, we conducted a  three-way complete factor ANOVA for each dependent variable. The experimental design comprised:
factor 1 ($L$): 3 levels (low/medium/high);
factor 2 (Homogeneity): 2 levels (CV=0.2/0.8); factor 3 (Density): 2 levels (sparse/dense).
For the approximation factor $\Delta$, we found significant effects from the $L$ factor ($F = 833.755$, $p< 0.001$) and the density factor ($F = 108.273$, $p < 0.001$). There was also an interaction effect between the factor $L$ and the density factor ($F = 127.480$, $p< 0.001$). A post hoc Tukey HSD test revealed statistically significant differences between the high $L$ and the low $L$ group (mean difference = $0.0756$, $p < 0.001$) and the medium $L$ group (mean difference = $0.0699$, $p < 0.001$). The difference between the medium $L$ group and the low $L$ group was smaller but still significant (mean difference = $0.0057$, $p< 0.05$).  For the density factor, Tukey's HSD test showed that the low density group was higher than the high density group (mean difference = $0.02049842$, $p< 0.001$). Figure \ref{fig:figInterDens} shows the pattern of interaction between both factors. The homogeneity factor has no significant effect on $\Delta$ ($F = 0.576$, $p > 0.05$).\par

\begin{figure}[ht]
   \centering
\includegraphics[width=0.7\textwidth]{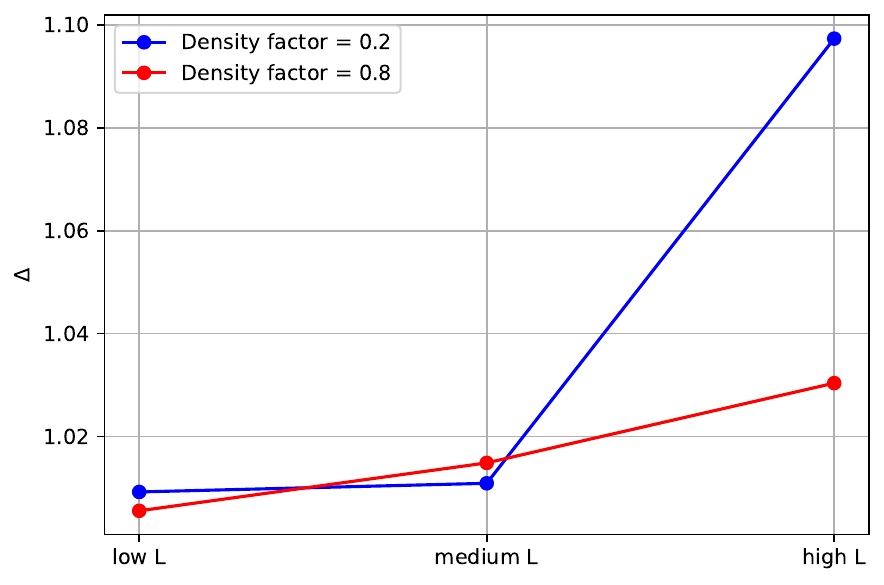}
    \caption{Interaction effect between $L$ and the density factor for the G2D algorithm.}
    \label{fig:figInterDens}
\end{figure}

For $\Delta_I$, only the factor $L$ showed significant effects on the approximation factor ($F = 349.040$, $p < 0.001$). A post-hoc Tukey’s HSD test indicated that the group of high $L$ has significantly greater mean of approximation factor than the group of low $L$  (mean difference = $0.0138$, $p < 0.001$), and also significantly greater than mean of medium $L$ (mean difference = $0.0129$, $p < 0.001$). However, the difference between medium $L$ and low $L$  was not statistically significant (mean difference = $0.0010$, $p > 0.05$). The density factor has no significant effects on $\Delta_I$ ($F =  2.8$, $p > 0.05$) as well as homogeneity factor ($F = 3.263$, $p > 0.05$). Consistent with our previous experimental results, we observed a statistically significant interaction effect between $L$ and segment homogeneity on algorithm performance. 
($F =12.949$, $p < 0.001$). Figure \ref{fig:figInterCVL} shows the interaction pattern between homogeneity and $L$.\par

\begin{figure}[ht]
   \centering
\includegraphics[width=0.7\textwidth]{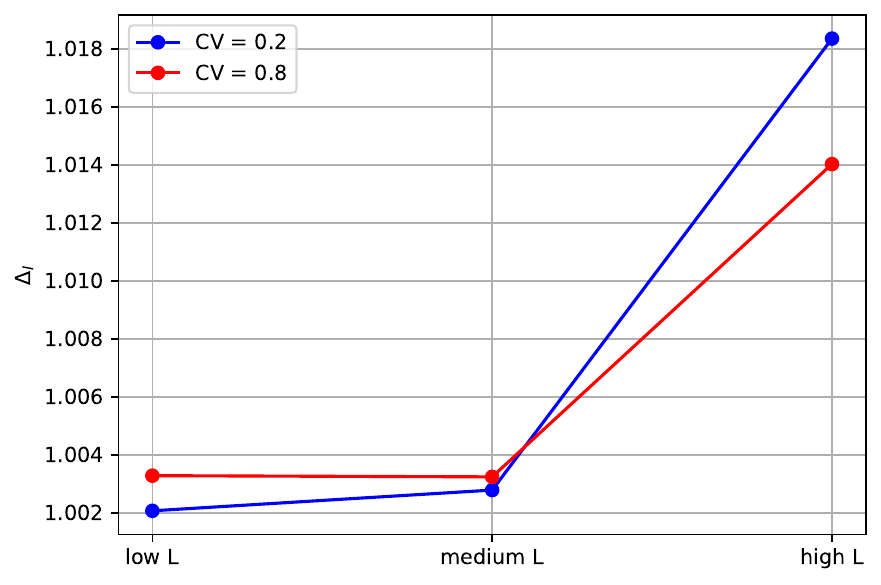}
    \caption{Interaction effect between $L$ and the homogeneity factor for the improved algorithm.}
    \label{fig:figInterCVL}
\end{figure}

The ANOVA test performed on the Gurobi™ time $T_G$ showed significant effects of the density factor ($F = 161.473$, $p < 0.001$)  as well as an interaction effect between the $L$ factor and the density factor ($F= 13.505$, $p< 0.001$). The homogeneity factor and the $L$ factor did not show any significant effect (homogeneity: $F = 0.028$. $L$-factor: $F = 2.780$). A post-hoc Tukey’s HSD test showed a statistically significant difference between the high density group and the low density group, with a mean difference of $264.82$ seconds ($p< 0.001$). In Figure \ref{fig:figInterGurobi} illustrates the interaction pattern between the factor $L$ and the density factor. We remark that the group with high density and low $L$ has a mean of $ 434.22$ seconds, with $\sigma = 1899.877$. \par

\begin{figure}[ht]
   \centering
\includegraphics[width=0.7\textwidth]{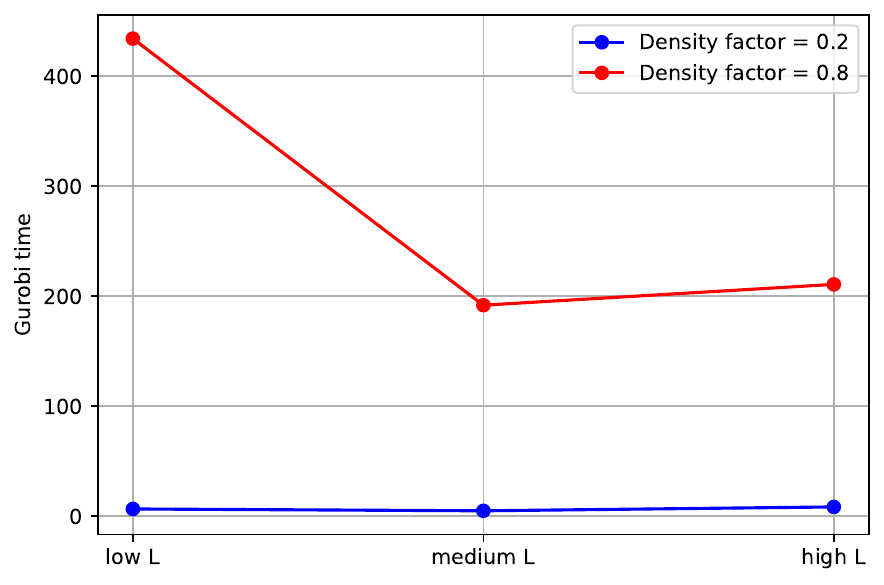}
    \caption{Interaction effect between $L$ and the density factor for the  Gurobi time in seconds.}
    \label{fig:figInterGurobi}
\end{figure}

Finally, we compare the approximation factor obtained empirically, $\Delta$, with the theoretical approximation factor exhibited in Theorem \ref{Theorem: Aproximation Theorem}. Recall that this last approximation factor is output sensitive. Let $D$ be the difference between the theoretical approximation factor and the empirical one.
A one-sample t-test revealed that the deviation metric $D$ was significantly greater than zero. 
($t = 43.246$, $p< 0.001$). The distribution of $D$ showed a mean of $0.233$ with a standard deviation of $0.310$. 

\subsection{Discussion of Experimental Findings}
Our experimental evaluation demonstrated consistently strong approximation performance under all test conditions.
Overall, the experiments showed low approximation factors with means no greater than $1.1$. In particular, the improved algorithm showed very low approximation factors with means lower than $1.05$, and extremely low approximation factor in the case of high homogeneity and low $L$ (with means around $1.03$), a common scenario in practical applications such as solar power plant inspections.  

The experimental results showed that the G2D-algorithm consistently outperforms its theoretical guarantees from Theorem \ref{Theorem: Aproximation Theorem}
(the difference mean was $0.2$, significantly greater than $0$). 
As another result, increasing the maximum tour length $L$ leads to a modest degradation in algorithm performance. However, this effect is significantly attenuated in the improved algorithm, as evidenced in
the second experiment, where only high $L$ has a difference with other groups.
The density factor is crucial in practical scenarios.
It was shown in the second experiment that high density improves the performance of G2D-algorithm, in particular in the cases of low $L$. 

Our benchmarking revealed significant computational limitations in the Gurobi™ solver's performance.
In particular, in cases of high density and low $L$, our computational experiments revealed significant limitations in the Gurobi™ solver's performance, reaching a maximum of time of several hours. Given that solar plants comprise hundreds of rows of panel segments, the Gurobi™ solver makes real-time deployment infeasible for field applications.
We observe that in such scenarios our algorithms have lower approximation factor. Due to computational constraints imposed by the Gurobi™ solver's exponential time complexity in high-density scenarios, we adopted a discretization level of $d=500$ in the second experiment.
However, $d=500$ provides sufficient resolution for algorithm comparison, homogeneity and density effects and
$L$-dependence analysis.

\par

 Finally, for illustrative purposes, Figure \ref{fig:exp4}
compares the solution provided by the G2D-algorithm with the solution obtained by digitizing the segments and applying the MILP formulation (solved with Gurobi™) for two instances, each consisting of only three segments. 
In the first example, shown on Figure\ref{fig:exp4} (a), the base station is located at the point $B=(0,-50)$ and the segments are $[-20,-13]$, $[-4, 10]$ and $[31, 60]$, on the line $y=0$.
G2D-algorithm produces a solution composed of two tours, $T_1 = (-20, 10)$ and $T_2 = (31, 60)$, with lengths, $l(T_1) = 134.84184321$ and $l(T_2) = L_{max} = 165.93276108$, respectively. The solution using Gurobi™  also consists of two tours, $T^*_1 = (-20, 32)$ and $T^*_2 = (32, 60)$, with lengths, $l(T^*_1) = 165.21493639$ and $l(T^*_2) = L^*_{max} =  165.46578508$, respectively. $L_{max} = 1.00282219L^*_{max}$. Note that Gurobi™'s solution spends time traversing the gap between the second and third segments, whereas the greedy solution returns to the base with sufficient remaining battery power. This is due to the fact that the set of tours used by our greedy procedure—optimal for a single drone—begins at the farthest point to be covered (on the right in this example) covering as much as possible. However, this approach does not yield the optimal solution when two drones are employed.

Something similar is happening in the second example, shown on Fig \ref{fig:exp4} (b). The base station is located at $B=(0,-50)$ and there are three segments, $[-4, 8]$, $[30, 38]$ and $[63, 79]$, on the line $y=0$. The G2D-algorithm again produces a solution composed of two tours: $T_1 = (-4, 8)$ and $T_2 = (30, 79)$, with lengths, $l(T_1) = 112.79570042$ and $l(T_2) = L_{max} =  200.80283422$. The Gurobi™'s solution  is also composed by two tours, $T^*_1 = (-4, 38)$ and $T^*_2 = (63, 79)$, with lengths, $l(T^*_1) = 154.96101869$ and $l(T^*_2) =L^*_{max} =  189.92340914$, respectively. $L_{max} = 1.05728322  L^*_{max}$.
Note that while the dashed tour (right) maximizes coverage from the farthest segment endpoint, this locally optimal strategy does not guarantee a globally optimal solution.
Another observation is that the optimal solution may require entering and exiting at interior points within a segment. These examples highlight the complexity of the optimization problem. However, in practice, the approximation algorithm is often favored in industry due to its simplicity in implementation and its ability to produce more geometrically intuitive routes.

\begin{figure}[ht]
\begin{tabular}{cc}

\includegraphics[width=0.4\textwidth, page=1]{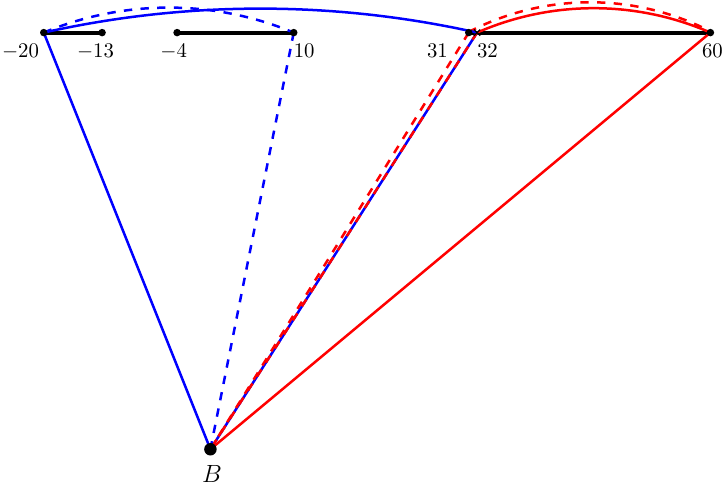}& \hspace{1cm}
\includegraphics[width=0.44\textwidth, page=2]{images/comparison_experiment.pdf} 
 \\
(a) & (b)
\end{tabular}
\caption{Examples of solutions provided G2D-algorithm in comparison with optimal solution. Dashed lines correspond to the greedy algorithm while solid lines represent the Gurobi™'s solution.}
\label{fig:exp4}
\end{figure}

\section{Conclusions and future research}\label{sec:conclu}
In this work, we have studied the problem of minimizing the flight time of multiple drones tasked with inspecting a set of segments positioned along a straight line.
We have proved that the problem is strongly NP-hard even for two drones.
Additionally, we have presented several polynomial-time algorithms to compute an approximation for a team of drones; we were able to prove that the approximation algorithms perform well not only in theory—due to their low constant-factor approximation—but also in simulations. By considering a discrete version of the problem, we have proposed a mixed-integer linear programming (MILP) formulation, using its solution as an optimal reference for the problem. 
Our algorithms have proven their ability to efficiently produce optimal solutions in a wide range of randomized scenarios. The approximation factors found in practice were even better than the theoretical ones. 
Moreover, we have identified specific scenarios where our algorithms excel. In particular we have shown that our algorithms perform nicely in practical scenarios as high density, high homogeneity or low $L$. 
Thus, we can conclude that the scenario where multiple drones must visit points along a line has been thoroughly analyzed, providing valuable insights and optimal solutions for various configurations.
These algorithms could be valuable not only for inspections in Concentrated Solar Power Plants but also for other industrial applications where targets are aligned along a straight line, such as high-tension power cables, oil and gas pipelines, or railway maintenance.

For future work, it would be interesting to develop approximation algorithms for scenarios where the targets are situated in more complex configurations, such as parallel lines or curved paths. Finally, the inclusion of multiple base stations introduces new optimization challenges of interest to the industry. Tours can start and end at different bases, potentially leading to more efficient inspection routes and further reducing total flight times.

\section*{Declarations}

\begin{itemize}
\item Conflict of interest: The authors declare that they have no conflict of interest. 
\item Ethical approval: No applicable.
\end{itemize}

\bibliography{eurocg24}
\end{document}